\documentclass[preprint]{imsart}

\RequirePackage{amssymb}
\RequirePackage{amsmath}
\RequirePackage{amsthm}
\RequirePackage{url}
\RequirePackage{bm}
\RequirePackage{color}
\RequirePackage[round]{natbib}
\RequirePackage[colorlinks=true,linkcolor=blue,citecolor=blue,pdfborder={0 0 0}]{hyperref}
\RequirePackage{grffile}
\RequirePackage{booktabs,rotating}

\startlocaldefs

\numberwithin{equation}{section}
\theoremstyle{plain}
\newtheorem{prop}{Proposition}[section]

\newtheorem{cor}[prop]{Corollary}
\newtheorem{lem}[prop]{Lemma}

\theoremstyle{remark}
\newtheorem{remark}[prop]{Remark}
\newtheorem{example}[prop]{Example}

\newcommand{\Ex}{\mathbb{E}}
\newcommand{\Cor}{\mathrm{cor}}
\newcommand{\R}{\mathbb{R}}
\newcommand{\N}{\mathbb{N}}
\newcommand{\Ic}{\mathcal{I}}
\newcommand{\Bc}{\mathcal{B}}
\newcommand{\Cc}{\mathcal{C}}
\newcommand{\Fc}{\mathcal{F}}
\newcommand{\Rc}{\mathcal{R}}
\newcommand{\Uc}{\mathcal{U}}
\newcommand{\Pc}{\mathcal{P}}
\newcommand{\Mc}{\mathcal{M}}

\newcommand{\Lc}{\mathcal{L}}
\newcommand{\Sc}{\mathcal{S}}
\newcommand{\Jc}{\mathcal{J}}

\renewcommand{\Pr}{\mathrm{P}}
\newcommand{\p}{\overset{\Pr}{\to}}

\newcommand{\vect}{\mathrm{vec}}
\newcommand{\diag}{\mathrm{diag}}
\newcommand{\supp}{\operatorname{supp}}
\newcommand{\1}{\mathbf{1}}
\newcommand{\eps}{\varepsilon}
\DeclareMathOperator*{\argsup}{arg\,sup}
\DeclareMathOperator*{\arginf}{arg\,inf}

\endlocaldefs

\allowdisplaybreaks

\begin{document}

\begin{frontmatter}

\title{Copula-like inference for discrete bivariate distributions with rectangular supports}
\runtitle{Copula-like inference for discrete bivariate distributions with rectangular supports}

\begin{aug}
  \author[I]{\fnms{Ivan} \snm{Kojadinovic}\ead[label=eI]{ivan.kojadinovic@univ-pau.fr}}
  \and
  \author[I,T]{\fnms{Tommaso} \snm{Martini}\ead[label=eT]{tommaso.martini@unito.it}}
  \address[I]{CNRS / Universit\'e de Pau et des Pays de l'Adour / E2S UPPA \\ Laboratoire de math\'ematiques et applications IPRA, UMR 5142 \\ B.P. 1155, 64013 Pau Cedex, France \\ \printead{eI}}
  \address[T]{Department of Mathematics G.\ Peano, Università degli Studi di Torino \\
    Via Carlo Alberto 10, 10123 Torino, Italy \\ \printead{eT}}

\end{aug}

\begin{abstract}
  After reviewing a large body of literature on the modeling of bivariate discrete distributions with finite support, \cite{Gee20} made a compelling case for the use of $I$-projections in the sense of \cite{Csi75} as a sound way to attempt to decompose a bivariate probability mass function (p.m.f.) into its two univariate margins and a bivariate p.m.f.\ with uniform margins playing the role of a discrete copula. From a practical perspective, the necessary $I$-projections on Fréchet classes can be carried out using the iterative proportional fitting procedure (IPFP), also known as Sinkhorn's algorithm or matrix scaling in the literature. After providing conditions under which a bivariate p.m.f.\ can be decomposed in the aforementioned sense, we investigate, for starting bivariate p.m.f.s with rectangular supports, nonparametric and parametric estimation procedures as well as goodness-of-fit tests for the underlying discrete copula. Related asymptotic results are provided and build upon a differentiability result for $I$-projections on Fréchet classes which can be of independent interest. Theoretical results are complemented by finite-sample experiments and a data example.
\end{abstract}

\begin{keyword}[class=MSC2010]
\kwd[Primary ]{62H17}
\kwd[; secondary ]{62G05, 62F03}
\end{keyword}

\begin{keyword}
 \kwd{asymptotic validity}
 \kwd{differentiability of $I$-pro\-ject\-ions on Fréchet classes}
 \kwd{goodness-of-fit tests}
 \kwd{iterative proportional fitting}
 \kwd{matrix scaling}
 \kwd{maximum pseudo-likelihood estimation}
 \kwd{method-of-moments estimation}
 \kwd{Sinkhorn's algorithm}
\end{keyword}

\tableofcontents

\end{frontmatter}


\section{Introduction}

Let $(X,Y)$ be a discrete random vector where $X$ (resp.\ $Y$) can take $r$ (resp.~$s$) distinct values for some strictly positive integers $r$ and $s$. As the approach proposed in this work depends only on the values of the probability mass function (p.m.f.) of $(X,Y)$, without loss of generality, we shall consider, for notational convenience only, that $(X,Y)$ takes its values in $I_{r,s} = I_r \times I_s$, where $I_r = \{1,\dots,r\}$ and $I_s = \{1,\dots,s\}$.

As is common in statistics, the p.m.f.\ $p$ of $(X,Y)$ is unknown and we instead have at hand $n$ (not necessarily independent) copies $(X_1,Y_1),\dots,(X_n,Y_n)$ of $(X,Y)$ that can be used to carry out inference on $p$. Classical approaches to model $p$ when $r$ and $s$ are not too large are log-linear models \cite[see, e.g.,][]{Agr13,Kat14,Rud18} and association models \citep[see, e.g.,][]{Goo85,Kat14}. We shall explore a different class of (semi-parametric or parametric) models in this work hinging upon the possibility of decomposing the unknown p.m.f.\ $p$ into its two univariate margins and a bivariate p.m.f.\ with uniform margins playing the role of a discrete copula. A strong case for such a decomposition was indeed recently made by \cite{Gee20} in a nice essay on dependence between discrete random variables. The underlying central ingredient is the concept of $I$-projection \citep[in the sense of][]{Csi75} on a Fréchet class of p.m.f.s as such projections turn out to preserve the \emph{odds ratio matrix} which encodes the dependence between $X$ and $Y$ in the considered discrete setting; see \citet[Section 2.4]{Agr13}, \citet[p 43]{Kat14}, \citet[p 123]{Rud18} or \citet[Section~4]{Gee23}. In practice, the $I$-projection of a p.m.f.\ on a Fréchet class can be carried out using the iterative proportional fitting procedure (IPFP), also known as Sinkhorn's algorithm or matrix scaling in the literature. The latter procedure takes the form of an algorithm (actually of several equivalent algorithms) whose aim is to adjust the elements of an input matrix so that it satisfies specified row and column sums. Since its use by \cite{Kru37} for the calculation of telephone traffic, it turned out to play a major role in a surprisingly large number of different scientific contexts \cite[see, e.g.,][for a nice review]{Ide16}.

A first contribution of this work, which can be of independent interest, is the statement of a result on the differentiability of an $I$-projection on a Fréchet class with respect to the starting p.m.f. The latter result will be necessary to study the asymptotics of the inference procedures proposed in the statistical part of this article. 

A second contribution \citep[related to the results given in][Section~6]{Gee20} is the statement of a corollary of known properties of $I$-projections which establishes that, under certain conditions on the bivariate p.m.f.\ $p$, the latter can be decomposed into its two univariate margins and a bivariate p.m.f.\ with uniform margins, called a \emph{copula p.m.f.}\ by \cite{Gee20}. Similar to the modeling of continuous multivariate distributions using copulas \citep[see, e.g.,][and the references therein]{HofKojMaeYan18} which exploits a well-known theorem of \cite{Skl59}, the obtained decomposition suggests to first separately model the margins and the copula p.m.f.\ and then ``glue'' the resulting estimates back together using an appropriate $I$-projection to obtain a parametric or a semi-parametric estimate of~$p$.

Building upon the aforementioned decomposition of bivariate p.m.f.s, a third contribution of this work is the study of nonparametric and parametric estimation procedures as well as goodness-of-fit tests for the underlying copula p.m.f. Note that these investigations are carried out under the assumption of strict positivity of the starting unknown p.m.f.~$p$. Indeed, $I$-projections may not always exist and, to guarantee the existence of the aforementioned copula-like decomposition of $p$, we shall thus additionally assume in the statistical part of this work that $p_{ij} > 0$ for all $(i,j) \in I_{r,s}$. Interestingly enough, the considered assumption could be regarded as the discrete analog of the assumption of strict positivity (on the interior of the unit square) of the copula density frequently made when modeling multivariate continuous distributions using copulas. Notice that, as shall be explained in more detail in our concluding remarks, the assumption of rectangular support for $p$ could be replaced by alternative conditions provided certain practical difficulties are dealt with. Fortunately, many applications seem to be compatible with the assumption of rectangular support.

This paper is organized as follows. In the second section, after reviewing the main properties of $I$-projections and the IPFP, we state differentiability results for $I$-projections on Fréchet classes. In Section~\ref{sec:sklar:like}, we provide conditions under which a copula-like decomposition of bivariate p.m.f.s based on $I$-projections on Fréchet classes exists. The fourth section is devoted to the nonparametric and parametric estimation of copula p.m.f.s and provides related asymptotic results as well as finite-sample findings based on simulations. Goodness-of-fit tests for copula p.m.f.s are studied next, both theoretically and empirically. A data example is provided in Section~\ref{sec:data:example} before concluding remarks are stated in the last section. All the proofs are relegated to a sequence of appendices.

\section{$I$-projections on Fréchet classes and the IPFP}

After introducing the notation, we recall the concept of $I$-projection due to \cite{Csi75} and its connection with the IPFP when the $I$-projection is carried out on a Fréchet class. We end this section with the statement of differentiability results for $I$-projections on Fréchet classes which might be of independent interest.

\subsection{Notation}

Let $\R^{r \times s}$ be the set of all $r \times s$ real matrices. Note the difference between $\R^{r \times s}$ and $\R^{rs}$, the latter denoting the set of all $rs$-dimensional real vectors. The element at row $i \in I_r = \{1,\dots,r\}$ and column $j \in I_s = \{1,\dots,s\}$ of a matrix $x \in \R^{r \times s}$ will be denoted by $x_{ij}$. Furthermore, for the row and column sums of~$x$, we use the classical notation
$$
x_{i+} = \sum_{j=1}^s x_{ij}, \qquad i \in I_r,  \qquad \text{and} \qquad x_{+j} = \sum_{i=1}^r x_{ij}, \qquad j \in I_s,
$$
respectively. Moreover, as we continue, let
\begin{multline}
  \label{eq:Gamma}
  \Gamma = \{x \in \R^{r \times s} : x_{ij} \geq 0, \, x_{i+} > 0, \,x_{+j} > 0 \text{ for all } (i,j) \in I_{r,s} \\ \text{ and }  \sum_{(i,j) \in I_{r,s}} x_{ij} = 1\}.
\end{multline}

In the context of this work, the set $\Gamma$ will equivalently represent the set of all bivariate p.m.f.s on $I_{r,s}$ whose univariate margins are strictly positive. As explained in the introduction,  since the forthcoming developments depend only on $r$, $s$ and p.m.f.\ values, restricting attention to p.m.f.s on $I_{r,s}$ is done without loss of generality. Indeed, any p.m.f.\ of interest defined on the Cartesian product of two sets of reals of cardinalities $r$ and $s$, respectively, can be ``relocated'' on~$I_{r,s}$.

As we continue, to distinguish bivariate p.m.f.s from other real matrices, we will always use lowercase letters for the former and uppercase letters for the latter. Also, given a bivariate p.m.f.\ $x$ in $\Gamma$, its first and second univariate margins will be denoted by $x^{[1]}$ and $x^{[2]}$, respectively. With our notation, we thus have
$$
x^{[1]}_i = x_{i+}, \qquad i \in I_r, \qquad \text{and} \qquad x^{[2]}_j = x_{+j}, \qquad j \in I_s.
$$

Next, let $a$ and $b$ be fixed univariate p.m.f.s on $I_r$ and $I_s$, respectively, such that $a_i > 0$ for all $i \in I_r$ and $b_j > 0$ for all $j \in I_s$ and let
\begin{equation}
  \label{eq:Gamma:a:.:b}
  \Gamma_{a,\cdot} = \{x \in \Gamma : x^{[1]} = a \} \quad (\text{resp.\ } \Gamma_{\cdot,b} = \{x \in \Gamma : x^{[2]} = b \})
\end{equation}
be the subset of $\Gamma$ in~\eqref{eq:Gamma} containing bivariate p.m.f.s whose first (resp.\ second) margin is $a$ (resp. $b$). Then
\begin{equation}
  \label{eq:Gamma:a:b}
  \Gamma_{a,b} = \{x \in \Gamma :  x^{[1]} = a \text{ and } x^{[2]} = b \} = \Gamma_{a,\cdot}\cap\Gamma_{\cdot,b}
\end{equation}
is the \emph{Fréchet class} of all bivariate p.m.f.s whose first margin is $a$ and second margin is $b$.

\subsection{$I$-projections on Fréchet classes}

To introduce the concept of $I$-projection due to \cite{Csi75}, we need to define the \emph{information divergence} (also known for instance as the \emph{Kullback--Leibler divergence} or \emph{relative entropy}) of a bivariate p.m.f.\ $y \in \Gamma$ with respect to a bivariate p.m.f.\ $x \in \Gamma$. The latter is defined by
\begin{equation}
\label{eq:KL:divergence}
D(y \| x) = \begin{cases}
\displaystyle \sum_{ (i,j) \in I_{r,s}} y_{ij} \log{\frac{y_{ij}}{x_{ij}}}  &\text{ if }\supp{(y)} \subset \supp{(x)},\\
\infty&\text{ otherwise,}
\end{cases}
\end{equation}
where, for any $x \in \Gamma$, $\supp(x) = \{(i,j) \in I_{r,s} : x_{ij} > 0 \}$, and with the convention that $0 \log 0 = 0$ and $0 \log \frac{0}{0} = 0$.

The following result is a direct consequence of Theorem 2.1 of \cite{Csi75} and the remark after Theorem 2.2 in the same reference.

\begin{prop}[Existence of $I$-projections]
\label{prop:I:proj}
  Let $x \in \Gamma$ and let $\Sc$ be a closed convex subset of $\Gamma$. Suppose furthermore that there exists a p.m.f.\ $y \in \Sc$ such that $\supp{(y)} \subset \supp{(x)}$. Then, there exists a unique $y^* \in \Sc$ such that  $y^*=\arginf_{y \in \Sc} D(y \|x)$. Moreover, we have that $\supp{(y)} \subset \supp{(y^*)}$ for every $y \in \Sc$ such that $\supp{(y)} \subset \supp{(x)}$.
\end{prop}

The bivariate p.m.f.\ $y^*$ in the statement of the previous proposition is usually referred to as the \emph{$I$-projection} of the starting bivariate p.m.f.\ $x$ on $\Sc$.

As explained in the introduction, the approach proposed by \cite{Gee20} relies on $I$-projections. Let $a$ and $b$ be fixed univariate p.m.f.s on $I_r$ and $I_s$, respectively, such that $a_i > 0$ for all $i \in I_r$ and $b_j > 0$ for all $j \in I_s$, and let $\Gamma_{a,b}$ in~\eqref{eq:Gamma:a:b} be the Fréchet class of all bivariate p.m.f.s whose first margin is $a$ and second margin is $b$. Note that $\Gamma_{a,b}$ is a closed and convex subset of $\Gamma$. As we continue, we shall denote the $I$-projection of a p.m.f.\ $x \in \Gamma$ on $\Gamma_{a,b}$, if it exists, by
\begin{equation}
  \label{eq:I:proj:Frechet}
  \Ic_{a,b}(x) = \arginf_{y \in \Gamma_{a, b}} D(y \| x).
\end{equation}

The following proposition, which, for instance, immediately follows from Corollary~3.3 in \cite{Csi75}, gives a more explicit form for the $I$-projection of $x \in \Gamma$ on the Fréchet classes $\Gamma_{a, b}$ when the latter contains a p.m.f.\ whose support is equal to that of $x$.
\begin{prop}
  \label{prop:I:proj:diag}
  Let $x \in \Gamma$. If there exists $y \in \Gamma_{a, b}$ such that $\supp{(y)} = \supp{(x)}$, then there exists two diagonal matrices $D_1 \in \R^{r \times r}$ and $D_2 \in \R^{s \times s}$ such that $\Ic_{a,b}(x) = y^* = D_1 \, x \, D_2$ and $\supp{(y^*)} = \supp{(x)}$. In this case, following \cite{Pre80}, $x$ and $y^*$ are said to be diagonally equivalent.
\end{prop}

As we can see from Propositions~\ref{prop:I:proj} and~\ref{prop:I:proj:diag}, the main practical issue before attempting to project a p.m.f.\ $x \in \Gamma$ on $\Gamma_{a,b}$ is thus to test for the existence of a bivariate p.m.f.\ $y \in \Gamma_{a, b}$ such that $\supp{(y)} \subset \supp{(x)}$. Necessary and sufficient conditions were for instance recently restated in Theorem~1 of \cite{BroLeu18} but seem to date back to Corollary 3 in \cite{Bac65}.

\begin{prop}[Testing for the existence of $I$-projections]
  \label{prop:cond:I:proj}
  Let $x \in \Gamma$ be a bivariate p.m.f.\ on $I_{r,s}$ and let $P_x$ be the corresponding probability measure on $I_{r,s}$. Furthermore, let $P_a$ (resp.\ $P_b$) be the probability measures on $I_r$ (resp.\ $I_s$) corresponding to the target univariate p.m.f.\ $a$ (resp.\ $b$). Then, there exists $y \in \Gamma_{a, b}$ such that $\supp{(y)} \subset \supp{(x)}$ if and only if
  \begin{equation}
    \label{eq:supp:cond}
    P_a(R) \leq  P_b(I_s \setminus C) \text{ for all } R \subset I_r \text{ and } C \subset I_s \text{ such that } P_x(R \times C) = 0,
  \end{equation}
and there exists $y \in \Gamma_{a, b}$ such that $\supp{(y)} = \supp{(x)}$ if and only if all the inequalities in~\eqref{eq:supp:cond} are strict.
\end{prop}

\subsection{The iterative proportional fitting procedure}
\label{sec:IPFP}

To carry out an $I$-projection on $\Gamma_{a,b}$ in practice, one can rely on the IPFP. The latter, also known as Sinkhorn's algorithm or matrix scaling in the literature, is a procedure whose aim is to adjust the elements of a matrix so that it satisfies specified row and column sums. We refer the reader for instance to \cite{Puk14}, \cite{Ide16} and \cite{BroLeu18} for a review of the procedure, its variations and its many contexts of application.

In principle, the IPFP could be applied to any input matrix in $\R^{r \times s}$ with nonnegative elements whose row and column sums are strictly positive. As this work is of a probabilistic nature, the IPFP will be described only for input matrices $x \in \Gamma$ that can be interpreted as bivariate p.m.f.s. In this context, the aim of the procedure is to adjust an input bivariate p.m.f.\ $x \in \Gamma$ so that it has $a$ as first margin and $b$ as second margin.

The IPFP with target margins $a$ and $b$ consists of applying repeatedly two transformations. The first one corresponds to the map $\Rc_a$ from $\Gamma$ in~\eqref{eq:Gamma} to $\Gamma_{a,\cdot}$ in~\eqref{eq:Gamma:a:.:b} defined by
\begin{equation}
  \label{eq:Rc:a}
\Rc_a(x) = \begin{bmatrix}
\frac{a_1 x_{11}}{x_{1+}}  & \dots & \frac{a_1 x_{1s}}{ x_{1+}} \\
\vdots &  & \vdots \\
\frac{a_r x_{r1}} {x_{r+}}  & \dots & \frac{a_r x_{rs}}{x_{r+}} \\
\end{bmatrix}, \qquad x \in \Gamma,
\end{equation}
while the second one corresponds to the map $\Cc_b$ from $\Gamma$ to $\Gamma_{\cdot,b}$ in~\eqref{eq:Gamma:a:.:b} defined by
\begin{equation}
  \label{eq:Cc:b}
\Cc_b(x) = \begin{bmatrix}
\frac{b_1 x_{11}}{ x_{+1}}  & \dots & \frac{b_s x_{1s}}{ x_{+s}} \\
\vdots &  & \vdots \\
\frac{b_1 x_{r1}}{ x_{+1}}  & \dots & \frac{b_s x_{rs}}{x_{+s}} \\
\end{bmatrix}, \qquad x \in \Gamma.
\end{equation}
Given an input matrix $x \in \Gamma$, the first (resp.\ second) map rescales the rows (resp.\ columns) of $x$ such that $\Rc_a(x)$ (resp.\ $\Cc_b(x)$) has first (resp.\ second) univariate margin equal to $a$ (resp.\ $b$).

Let $N \in \N$. The $(2N+1)$th step of the IPFP with target margins $a$ and $b$ can then be expressed as
\begin{equation}
  \label{eq:2N+1step:IPFP}
\Ic_{2N+1,a,b}(x) =  \Rc_a \circ (\Cc_b \circ \Rc_a)^{(N)}(x), \qquad x \in \Gamma,
\end{equation}
while the $2N$th step, $N \geq 1$, can be expressed as
\begin{equation}
  \label{eq:2Nstep:IPFP}
\Ic_{2N,a,b}(x) = (\Cc_b \circ \Rc_a)^{(N)}(x), \qquad x \in \Gamma,
\end{equation}
where the superscript $(N)$ denotes composition $N$ times and composition 0 times is taken to be the identity. We will say that the IPFP converges for a starting bivariate p.m.f.\ $x \in \Gamma$ if the sequence $\{ \Ic_{N,a,b}(x) \}_{N \geq 1}$ converges. 

In practice, if the IPFP of the bivariate p.m.f.\ $x$ converges, for some user-defined $\eps > 0$, $\lim_{N \rightarrow \infty} \Ic_{N,a,b}(x)$ is approximated by $\Ic_{N,a,b}(x)$, where $N$ (which depends on $x$) is an integer such that
\begin{equation}
  \label{eq:IPFP:criterion}
\| \Ic_{N,a,b}(x) - \Ic_{N-1,a,b}(x) \|_{1,1} \leq \eps,
\end{equation}
where $\|x\|_{1,1} = \sum_{(i,j) \in I_{r,s}} |x_{ij}|$. Results on the number of iterations required for the previous condition to be satisfied are discussed for instance in \cite{KalLarRicSim08}, \cite{Puk14} and \cite{ChaKha21}.

The next proposition, which, for instance, immediately follows from \citet[Theorem~3]{Puk14} and \citet[Theorems~2 and~3]{BroLeu18}, gives the connection between $I$-projections on Fréchet classes and the IPFP.

\begin{prop}(Link between $I$-projections on Fréchet classes and the IPFP)
  \label{prop:conv:IPFP}
  Let $x \in \Gamma$. The sequence $\{ \Ic_{N,a,b}(x) \}_{N \geq 1}$ converges if and only if there exists $y \in \Gamma_{a, b}$ such that $\supp{(y)} \subset \supp{(x)}$ and, if $\lim_{N \rightarrow \infty} \Ic_{N,a,b}(x)$ exists, it is equal to $\Ic_{a,b}(x)$ in~\eqref{eq:I:proj:Frechet}.
\end{prop}

In other words, the IPFP of a bivariate p.m.f.\ $x$ converges (to the $I$-projection of $x$ on the corresponding Fréchet class) if and only if the $I$-projection of $x$ on the corresponding Fréchet class exists.

\subsection{Differentiability results for $I$-projections on Fréchet classes}

Theorem 3.3 in \citet{GieRef17} implies a form of continuity of an $I$-projection on a Fréchet class when it exists (see Lemma~\ref{lem:continuity} in Appendix~\ref{proof:prop:diff:I:proj} for a precise statement). As shall become clearer in Section~\ref{sec:est}, to investigate the asymptotics of the statistical inference procedures proposed later in this work, we will also need an $I$-projection on a Fréchet class, when it exists, to be differentiable in a related sense. The latter property does not seem to have been investigated in the literature although it can be connected to the work \cite{JimPinAlbMor11} as explained in Remark~\ref{rem:JimEtAl} below.

The aim of this section is to state a differentiability result, roughly speaking, with respect to certain strictly positive submatrices of the input matrix. To make the statement precise, we first need to introduce additional notation.

For any $S \subset I_{r,s}$, $S \neq \emptyset$, let $\vect_S$ be the map from $\R^{r \times s}$ to $\R^{|S|}$ which, given a matrix $y \in \R^{r \times s}$, returns a column-major vectorization of $y$ ignoring the elements $y_{ij}$ such that $(i,j) \not \in S$.  With some abuse of notation, we can write
\begin{equation}
\label{eq:vect:S}
\vect_S(y) = ( y_{ij} )_{(i,j) \in S}, \qquad y \in \R^{r \times s}.
\end{equation}
Let $\vect_S^{-1}$ be the map from $\R^{|S|}$ to $\R^{r \times s}$ which, given a vector $v$ in $[0,1]^{|S|}$ expressed (with some abuse of notation) as $( v_{ij} )_{(i,j) \in S}$, returns a matrix in $\R^{r \times s}$ whose element at position $(i,j) \in S$ is equal to $v_{ij}$ and whose element at position $(i,j) \not \in S$ is left unspecified. Then, for any two matrices $y,y' \in \R^{r \times s}$, write $y \overset{S}{=} y'$ if $y_{ij} = y'_{ij}$ for all $(i,j) \in S$. Furthermore, recall the definitions of $\Gamma$ in~\eqref{eq:Gamma} and $\Gamma_{a,b}$ in~\eqref{eq:Gamma:a:b}, and, for any $A,B \subset T \subset I_{r,s}$, $A,B \neq \emptyset$, let
\begin{align}
  \label{eq:Gamma:T}
  \Gamma_T &= \{y \in \Gamma : \supp{(y)} = T \}, \\
  \label{eq:Lambda:B:T}
  \Lambda_{B,T} &= \{w \in (0,1)^{|B|} : \text{ there exists } y \in \Gamma_T \text{ s.t.\ } \vect_B^{-1}(w) \overset{B}{=} y  \}, \\
  \label{eq:Gamma:ab:T}
  \Gamma_{a,b,T} &= \{y \in \Gamma_{a,b} : \supp{(y)} = T \}, \\
  \label{eq:Lambda:ab:A:T}
  \Lambda_{a,b,A,T} &= \{z \in (0,1)^{|A|} : \text{ there exists } y \in \Gamma_{a,b,T} \text{ s.t.\ } \vect_A^{-1}(z) \overset{A}{=} y  \}.
\end{align}
Finally, for any function $G$ from $\R^k \times \R^m$ to $\R^l$ differentiable at $(u,v) \in \R^k \times \R^m$, we shall use the notation
$$
\partial_1 G(u, v) = \begin{bmatrix}
\frac{\partial G_1(w,v)}{\partial w_1} \Big |_{w = u} & \dots & \frac{\partial G_1(w,v)}{\partial w_k} \Big |_{w = u} \\
\vdots &  & \vdots \\
\frac{\partial G_l(w,v)}{\partial w_1} \Big |_{w = u} & \dots & \frac{\partial G_l(w,v)}{\partial w_k} \Big |_{w = u} \\
\end{bmatrix}
$$
and
$$
\partial_2 G(u, v) = \begin{bmatrix}
\frac{\partial G_1(u,w)}{\partial w_1} \Big |_{w = v} & \dots & \frac{\partial G_1(u,w)}{\partial w_m} \Big |_{w = v} \\
\vdots &  & \vdots \\
\frac{\partial G_l(u,w)}{\partial w_1} \Big |_{w = v} & \dots & \frac{\partial G_l(u,w)}{\partial w_m} \Big |_{w = v} \\
\end{bmatrix}.
$$

The following proposition, proven in Appendix~\ref{proof:prop:diff:I:proj}, is notationally complex but hinges on the following simple fact: a bivariate p.m.f.\ with support $T$ can be expressed in terms of $|B| = |T| - 1$ elements and a p.m.f.\ with support $T$ and univariate margins $a$ and $b$ can be expressed in terms of $|A| < |B|$ elements so that the $I$-projection of a p.m.f.\ in $\Gamma_T$ on $\Gamma_{a,b,T}$ can be regarded as a map from an open subset of $(0,1)^{|B|}$ to $(0,1)^{|A|}$. Looking at the aforementioned $I$-projection in such a way allows us to consider it as a map related to an unconstrained optimization problem and, eventually, to apply the implicit function theorem \citep[see, e.g.,][Theorem 17.6, p 450]{Fit09}. 

\begin{prop}
\label{prop:diff:I:proj}
Let $T \subset I_{r,s}$, $|T| > 1$, such that $\Gamma_{a,b,T}$ in \eqref{eq:Gamma:ab:T} is nonempty, and assume that there exists a nonempty subset $A \subsetneq T$ such that:
\begin{enumerate}
\item for any $y \in \Gamma_{a,b,T}$, the elements $y_{ij}$, $(i,j) \in T \setminus A$, can be recovered from the elements $y_{ij}$, $(i,j) \in A$, using the $r + s$ constraints $y_{i+} = a_i$, $i \in I_r$, and $y_{+j} = b_j$, $j \in I_s$,
\item the set $\Lambda_{a,b,A,T}$ in~\eqref{eq:Lambda:ab:A:T} is an open subset of $\R^{|A|}$.
\end{enumerate}
Let $B$ be any subset of $T$ such that $|B| = |T| - 1$ and let $d$ be the one-to-one map from $\Lambda_{B,T}$ in~\eqref{eq:Lambda:B:T} to $\Gamma_T$ in~\eqref{eq:Gamma:T} which, given $w \in \Lambda_{B,T}$, returns the unique $y \in \Gamma_T$ obtained by recovering the only unspecified element of $\vect_B^{-1}(w) \in \Gamma_T$ using the constraint $\sum_{(i,j) \in T} y_{ij} = 1$. Then, the map $\vect_A \circ \Ic_{a,b} \circ d$ from $\Lambda_{B,T}$ to $\Lambda_{a,b,A,T}$, where $\Ic_{a,b}$ is defined in~\eqref{eq:I:proj:Frechet}, is differentiable at any $\vect_B(x)$, $x \in \Gamma_T$, with Jacobian matrix at $\vect_B(x)$ equal to
\begin{equation*}
  \label{eq:Jacobian}
- [\partial_1 \partial_1 H(\vect_A(\Ic_{a,b}(x)) \| \vect_B(x)) ]^{-1} \partial_2  \partial_1 H(\vect_A(\Ic_{a,b}(x)) \| \vect_B(x)).
\end{equation*}
In the above centered display, $H$ is the map from $\Lambda_{a,b,A,T} \times \Lambda_{B,T}$ to $[0,\infty)$ defined by
\begin{equation}
  \label{eq:H:def}
H(z \| w ) = D(c(z) || d(w)), \qquad z \in \Lambda_{a,b,A,T}, w \in \Lambda_{B,T},
\end{equation}
where $D$ is defined in~\eqref{eq:KL:divergence} and $c$ is the one-to-one map from $\Lambda_{a,b,A,T}$ in~\eqref{eq:Lambda:ab:A:T} to $\Gamma_{a,b,T}$ in \eqref{eq:Gamma:ab:T} which, given $z \in \Lambda_{a,b,A,T}$, returns the unique $y \in \Gamma_{a,b,T}$ obtained by recovering the unspecified elements of $\vect_A^{-1}(z) \in \Gamma_{a,b,T}$ using the $r + s$ constraints $y_{i+} = a_i$, $i \in I_r$ and $y_{+j} = b_j$, $j \in I_s$.
\end{prop}

\begin{remark}
  \label{rem:JimEtAl}
  Following the suggestions of a Referee, while revising this work, we explored the connections of the statistical results to be stated in Section~\ref{sec:est} with \emph{minimum divergence estimation} (see Remark~\ref{rem:min:div:est} below for more details) and found that Lemma~1 in \cite{JimPinAlbMor11} is connected with Proposition~\ref{prop:diff:I:proj} above. Some additional thinking reveals however that the aforementioned lemma cannot be used to obtain a differentiability result for an $I$-projection on a Fréchet class. Furthermore, its proof seems incomplete as the assumptions considered in \cite{JimPinAlbMor11} do not seem to guarantee a key matrix invertibility required to apply the implicit function theorem \citep[see~(17.17) in][Theorem 17.6, p 450]{Fit09}. \qed
\end{remark}

\section{Copula-like decomposition of bivariate p.m.f.s}
\label{sec:sklar:like}

Using the results given in the previous section, we shall now state a copula-like decomposition for a bivariate p.m.f.\ on $I_{r,s}$ and strictly positive univariate margins. Let $u^{[1]}$ (resp.\ $u^{[2]}$) be the univariate p.m.f.\ of the uniform distribution on $I_r$ (resp.\ $I_s$) and let $\Gamma_{\text{unif}} = \Gamma_{u^{[1]},u^{[2]}}$ be the Fréchet class of all bivariate p.m.f.s on $I_{r,s}$ whose univariate margins are $u^{[1]}$ and $u^{[2]}$, respectively. The $I$-projection of a p.m.f.\ $x \in \Gamma$ on $\Gamma_{\text{unif}}$, if it exists, will be denoted by
\begin{equation}
  \label{eq:I:proj:unif}
  \Uc(x) = \Ic_{u^{[1]},u^{[2]}}(x),
\end{equation}
where $\Ic_{u^{[1]},u^{[2]}}$ is defined as in~\eqref{eq:I:proj:Frechet} with $a=u^{[1]}$ and $b=u^{[2]}$.

The following result, proven in Appendix~\ref{proof:prop:sklar:like}, will play, in the next section, a role analog to Sklar's theorem when modeling continuous multivariate distributions using copulas.

\begin{prop}[Copula-like decomposition of bivariate p.m.f.s]
\label{prop:sklar:like}
Let $p$ be a bivariate p.m.f.\ on $I_{r,s}$ with strictly positive univariate margins $p^{[1]}$ and $p^{[2]}$. Then, the following two assertions are equivalent:
\begin{enumerate}
\item There exists a bivariate p.m.f.\ $v$ on $I_{r,s}$ with uniform margins such that $\supp{(v)} = \supp{(p)}$.
\item There exists a unique bivariate p.m.f.\ $u$ on $I_{r,s}$ with uniform margins such that
\begin{equation}
  \label{eq:sklar:like}
  p = \Ic_{p^{[1]},p^{[2]}}(u).
\end{equation}
Furthermore, the unique bivariate p.m.f.\ $u$ on $I_{r,s}$ with uniform margins in~\eqref{eq:sklar:like} is given by $u = \Uc(p)$, where $\Uc$ is defined in~\eqref{eq:I:proj:unif}.
\end{enumerate}
\end{prop}

  \begin{example}
In the setting of the previous proposition, let us illustrate the fact that, if Assertion~1 does not hold, then~\eqref{eq:sklar:like} does not hold. Assume that $p$ is the p.m.f.\ on $I_{3,3}$ represented by the matrix
$$
p=\begin{bmatrix}
  \frac{1}{7} & \frac{1}{7} & \frac{1}{7} \\
  \frac{1}{7} & 0 & \frac{1}{7} \\
  \frac{1}{7} & 0 & \frac{1}{7} \\
\end{bmatrix}.
$$
With $u^{[1]} = u^{[2]}$ the uniform p.m.f.\ on $I_3$, note that~\eqref{eq:supp:cond} is satisfied for $x = p$, $a = u^{[1]}$ and~$b = u^{[2]}$ with equality holding for $R=\{1, 3\}$ and $C=\{2\}$. Therefore, according to Proposition~\ref{prop:cond:I:proj}, there exists $v \in \Gamma_{\text{unif}}$ with $\supp{(v)} \subset \supp{(p)}$, but there does not exist any $v \in \Gamma_{\text{unif}}$ such that $\supp{(v)} = \supp{(p)}$. Hence, from Proposition~\ref{prop:I:proj}, we know that $u = \Uc(p)$ exists. From Proposition~\ref{prop:conv:IPFP}, we can compute $u$ by using the IPFP via~\eqref{eq:2N+1step:IPFP} and~\eqref{eq:2Nstep:IPFP}. It can be verified that, for any $N \in \N$ (that is, after a row scaling step of the IPFP),
$$
\Rc_{u^{[1]}} \circ (\Cc_{u^{[2]}} \circ \Rc_{u^{[1]}})^{(2N)}(p) = \frac{1}{3}\begin{bmatrix}
  \frac{1}{3(N+1)} & \frac{1+3N}{3(N+1)} & \frac{1}{3(N+1)} \\
  \frac{1}{2} & 0 & \frac{1}{2} \\
  \frac{1}{2} & 0 & \frac{1}{2} \\
\end{bmatrix},
$$
where $\Rc_{u^{[1]}}$ and $\Cc_{u^{[2]}}$ are defined as in~\eqref{eq:Rc:a} and~\eqref{eq:Cc:b}, respectively, with $a = b = u^{[1]} = u^{[2]}$, and that, for any $N > 0$ (that is, after a column scaling step of the IPFP),
$$
(\Cc_{u^{[2]}} \circ \Rc_{u^{[1]}})^{(2N)}(p) = \frac{1}{3}\begin{bmatrix}
  \frac{1}{1 + 3N} & 1 & \frac{1}{1+3N} \\
  \frac{3N}{2(1+3N)} & 0 & \frac{3N}{2(1+3N)} \\
  \frac{3N}{2(1+3N)} & 0 & \frac{3N}{2(1+3N)} \\
\end{bmatrix}.
$$
Letting $N$ tend to $\infty$ for both sequences, we obtain
$$
u = \Uc(p) = \begin{bmatrix}
  0 & \frac{1}{3} & 0 \\
  \frac{1}{6} & 0 & \frac{1}{6} \\
  \frac{1}{6} & 0 & \frac{1}{6} \\
\end{bmatrix}
$$
and we see that $\supp{(u)} \subsetneq \supp{(p)}$. Furthermore, $\Ic_{p^{[1]},p^{[2]}}(u)$ does not exist as a consequence of Proposition~\ref{prop:cond:I:proj} since~\eqref{eq:supp:cond} with $x = u$, $a = p^{[1]}$ and~$b = p^{[2]}$ is contradicted for $R=\{1\}$ and $C=\{1,3\}$. Thus, the decomposition in~\eqref{eq:sklar:like} does not hold.  \qed
\end{example}

When it exists, the unique p.m.f.\ on $I_{r,s}$ with uniform margins $u$ in~\eqref{eq:sklar:like} was called the \emph{copula p.m.f.}\ of $p$ by \citet{Gee20}. Akin to Sklar's theorem, we see that in~\eqref{eq:sklar:like} the copula p.m.f.\ $u$ of $p$ and the marginal p.m.f.s $p^{[1]}$ and $p^{[2]}$ are ``glued together'' via an $I$-projection (on $\Gamma_{p^{[1]},p^{[2]}}$) to obtain the bivariate p.m.f.~$p$. As a consequence of Proposition~\ref{prop:I:proj:diag}, $u$ and $p$ share the same \emph{odds ratio matrix}; see \citet[Section 2.4]{Agr13}, \citet[p 43]{Kat14} or \citet[p 123]{Rud18}. According to~\citet[Section 6]{Gee20} \citep[see also][Section~4]{Gee23}, the latter indicates that both bivariate p.m.f.s share the same dependence properties. It can actually be verified, again from Proposition~\ref{prop:I:proj:diag}, that any bivariate p.m.f.\ $y \in \Gamma_{a, b}$ such that $\supp{(y)} = \supp{(u)}$ (assuming that it exists) obtained by an $I$-projection of $u$ on $\Gamma_{a, b}$ would keep the odds ratio matrix constant. As explained in \cite{Gee20}, the unique copula p.m.f.\ $u$ of $p$ is merely a natural (yet arbitrarily chosen) representative of the equivalence class of all bivariate p.m.f.s with the same odds ratio matrix, that is, with the same dependence properties.

\begin{remark}
  In \citet[Definition 6.1]{Gee20}, copula p.m.f.s were defined as bivariate p.m.f.s on $\{(i/(r+1), j/(s+1)):(i,j) \in I_{r,s}\} \subset [0,1]^2$, by analogy with bivariate copulas which are defined on $[0,1]^2$. We have chosen to define them as bivariate p.m.f.s on $I_{r,s}$ for simplicity. \qed
\end{remark}

In the same way that a continuous bivariate distribution can be summarized in terms of dependence by a moment of the underlying copula \cite[such as Spearman's rho or Kendall's tau -- see, e.g.,][Chapter~2 and the references therein]{HofKojMaeYan18}, when the decomposition in~\eqref{eq:sklar:like} holds, a bivariate p.m.f.\ $p$ can be summarized by a moment of its copula p.m.f.\ $u$. In the spirit of Spearman's rho, as a possible summary of $u$, \cite{Gee20} proposed to consider Pearson's linear correlation $\Cor(U,V)$ when $(U,V)$ has p.m.f.\ $u$. Because of the analogy with Yule's colligation coefficient \citep{Yul12} for 2 by 2 contingency tables, \citet[Section 6.6]{Gee20} suggested to call the resulting quantity \emph{Yule's coefficient}. Specifically, let $\Upsilon$ be the map from $\Gamma_{\text{unif}}$ to $[-1,1]$ defined, for any $v \in \Gamma_{\text{unif}}$, by
\begin{equation}
  \label{eq:Upsilon}
\Upsilon(v) = 3 \sqrt{\frac{(r-1)(s-1)}{(r+1)(s+1)}} \left(\frac{4}{(r-1)(s-1)} \sum_{(i,j) \in I_{r,s}} (i-1)(j-1) \, v_{ij} - 1 \right).
\end{equation}
Yule's coefficient of $p$ (or $u$) as proposed by \cite{Gee20} is then simply
\begin{equation}
  \label{eq:Yule}
  \rho = \Cor(U,V) =  \Upsilon \circ \Uc(p) =  \Upsilon(u).
\end{equation}
Recall that $p$ is the p.m.f.\ of $(X,Y)$. It is important to note that $\rho$ does not coincide with $\Cor(X,Y)$ in general.

Many other moments of $u$ could be used. As possible alternatives to $\rho$, we shall consider coefficients based on Goodman's and Kruskal's gamma \citep{GooKru54} and Kendall's tau $b$ \citep[see, e.g.,][]{KenGib90} of $(U,V)$ when $(U,V)$ has p.m.f.\ $u$. The following proposition, proven in Appendix~\ref{proof:prop:sklar:like}, gives the expressions of Goodman's and Kruskal's gamma and Kendall's tau $b$ for bivariate p.m.f.s on $I_{r,s}$ with uniform margins.

\begin{prop}
  \label{prop:G:T}
For any p.m.f.\ $v \in \Gamma_{\text{unif}}$, Goodman's and Kruskal's gamma of $v$ can be expressed as
\begin{equation}
  \label{eq:G}
  G(v) = \frac{2\kappa(v) - 1 + 1/r + 1/s - \|v\|_2^2}{1 - 1/r - 1/s + \|v\|_2^2},
\end{equation}
and Kendall’s tau b of $v$ can be expressed as
\begin{equation}
  \label{eq:T}
  T(v) = \frac{\sqrt{r s} \{2\kappa(v) - 1 + 1/r + 1/s - \|v\|_2^2\}}{\sqrt{(r-1)(s-1)}},
\end{equation}
where
\begin{align*}
  \kappa(v) & = 2 \sum_{\substack{i \in \{2,\dots,r\} \\ j \in \{1,\dots,s-1\}}} \sum_{\substack{i' \in \{1,\dots,i-1\} \\ j' \in \{j+1,\dots,s\}}} v_{ij} v_{i'j'}. 
\end{align*}
\end{prop}

By analogy with Yule's coefficient in~\eqref{eq:Yule}, we define the \emph{gamma coefficient} of $p$ (or $u$) and the \emph{tau coefficient} of $p$ (or $u$) to be
\begin{equation}
  \label{eq:gamma:tau}
  \gamma = G \circ \Uc(p) =  G(u) \qquad \text{and} \qquad \tau = T \circ \Uc(p) =  T(u),
\end{equation}
respectively, where the maps $G$ and $T$ are defined in~\eqref{eq:G} and~\eqref{eq:T}, respectively. It is important to keep in mind that these coefficients coincide with Goodman's and Kruskal's gamma and Kendall's tau $b$ of $p$ only when $p$ has uniform margins.

It is easy to verify that all three coefficients $\rho$, $\gamma$ and $\tau$ are equal to $1$ (resp.\ $-1$) when $r = s$ and the copula p.m.f.\ $u$ is a diagonal (resp.\ anti-diagonal) matrix. As explained in \citet[Sections 5 and~6]{Gee20}, in this case, $u$ corresponds to the upper (resp.\ lower) Fréchet bound for a square copula p.m.f.

\section{Estimation of copula p.m.f.s}
\label{sec:est}

Let us go back to the setting considered in the introduction. As explained therein, we are interested in modeling the bivariate p.m.f.\ $p$ of a discrete random vector $(X,Y)$ with (rectangular) support $I_{r,s}$, that is, $\supp{(p)} = I_{r,s}$. To do so, we have at our disposal $n$ (not necessarily independent) copies $(X_1,Y_1)$, \dots, $(X_n,Y_n)$ of $(X,Y)$. Note that there exists a trivial bivariate p.m.f.\ with uniform margins with support equal to $I_{r,s}$: it is the independence copula p.m.f.\ $\pi$ defined by $\pi_{ij} = 1 / (rs)$, $(i,j) \in I_{r,s}$. We can therefore immediately apply Proposition~\ref{prop:sklar:like} to obtain the decomposition of $p$  given in~\eqref{eq:sklar:like}.

By analogy with estimation approaches classically used in the context of copula modeling for continuous random variables \citep[see, e.g.,][Chapter 4 and the references therein]{HofKojMaeYan18}, \eqref{eq:sklar:like} suggests to proceed in three steps to form a parametric or semi-parametric estimate of $p$. For instance, to obtain a parametric estimate, one could proceed as follows:
\begin{enumerate}
\item Estimate the univariate margins $p^{[1]}$ and $p^{[2]}$ of $p$ parametrically; let $p^{[1,\alpha^{[n]}]}$ and $p^{[2,\beta^{[n]}]}$ be the resulting estimates.
\item Estimate the copula p.m.f.\ $u$ parametrically; let $u^{[\theta^{[n]}]}$ be the resulting estimate.
\item Form  a parametric estimate of $p$ in~\eqref{eq:sklar:like} via an $I$-projection as
  $$
  p^{[\alpha^{[n]},\beta^{[n]},\theta^{[n]}]} = \Ic_{p^{[1,\alpha^{[n]}]}, p^{[2,\beta^{[n]}]}}(u^{[\theta^{[n]}]}).
  $$
\end{enumerate}

\begin{remark}
  With the notation of Section \ref{sec:IPFP} and given $x \in \Gamma$ with $\supp{(x)} = I_{r,s}$, Theorem 6.2 of \citet{GieRef17} implies that the IPFP of $x$ depends continuously on $x$ and on the underlying target marginal p.m.f.s $a$ and~$b$. Under a natural assumption of strict positivity of the parametric model for $u$ (see Section~\ref{sec:par:est} below), from Proposition~\ref{prop:conv:IPFP} and the continuous mapping theorem, this implies that $p^{[\alpha^{[n]},\beta^{[n]},\theta^{[n]}]} = \Ic_{p^{[1,\alpha^{[n]}]}, p^{[2,\beta^{[n]}]}}(u^{[\theta^{[n]}]})$ is a consistent estimator of $p$ if $p^{[1,\alpha^{[n]}]}$, $p^{[2,\beta^{[n]}]}$ and $u^{[\theta^{[n]}]}$ are consistent estimators of $p^{[1]}$, $p^{[2]}$ and $u$, respectively. \qed
\end{remark}

The first estimation step above can be carried out using classical approaches in statistics. The aim of this section is to address the second step. However, as the parametric estimation of $u$ will turn out to be strongly related to a natural nonparametric estimator of $u$, we first investigate the latter. As we continue, all convergences are as $n \to \infty$ unless otherwise stated.

\subsection{Nonparametric estimation}

A straightforward approach to obtain a nonparametric estimator of $u$ is to use the plugin principle. Since a natural estimator of $p$ is $\hat p^{[n]}$ defined by
\begin{equation}
  \label{eq:hat:pn}
  \hat p^{[n]}_{ij} = \frac{1}{n} \sum_{k=1}^n \1(X_k = i, Y_k = j), \qquad (i,j) \in I_{r,s},
\end{equation}
a meaningful estimator of $u$ would simply be $\Uc(\hat p^{[n]})$, where $\Uc$ is defined in~\eqref{eq:I:proj:unif}. However, some thought reveals that this estimator may not always exist when $n$ is small. Indeed, the fact that $\supp{(p)} = I_{r,s}$ obviously does not imply that $\supp{(\hat p^{[n]})} = I_{r,s}$ for all $n$. Furthermore, there is no guarantee that there exists a bivariate p.m.f.\ $v^{[n]}$ on $I_{r,s}$ with uniform margins such that $\supp{(v^{[n]})} \subset \supp{(\hat p^{[n]})}$, which is necessary to ensure that $\Uc(\hat p^{[n]})$ exists. A way to solve this issue is to consider a smooth version of $\hat p^{[n]}$ in~\eqref{eq:hat:pn}. Although many solutions could be considered \citep[see, e.g.,][]{Sim95}, we opt for one of the simplest ones and consider the estimator $p^{[n]}$ of $p$ defined by
\begin{equation}
  \label{eq:pn}
  p^{[n]}_{ij} = \frac{1}{n+1} \left\{ \sum_{k=1}^n \1(X_k = i, Y_k = j) + q_{ij} \right\}, \qquad (i,j) \in I_{r,s},
\end{equation}
where $q$ is a p.m.f.\ on $I_{r,s}$ with $\supp{(q)} = I_{r,s}$. The latter can be equivalently rewritten in terms of $\hat p^{[n]}$ in~\eqref{eq:hat:pn} and $q$ as
\begin{equation}
  \label{eq:pn:q}
p^{[n]} = \frac{n}{n+1} \hat p^{[n]} + \frac{1}{n+1} q.
\end{equation}
In our numerical experiments, we considered two possibilities for $q$: the independence copula p.m.f.\ $\pi$ and, following the suggestion of a Referee, the p.m.f.\ $\hat p^{[n,1]} \hat p^{[n,2], \top}$ which has copula p.m.f.\ $\pi$ and the same margins as $\hat p^{[n]}$ in~\eqref{eq:hat:pn}. Note that the latter choice is only meaningful when both $\hat p^{[n,1]}$ and $\hat p^{[n,2]}$ are strictly positive (of course, if this is not the case, a practitioner could decide to reduce the cardinalities $r$ or $s$ of the marginal supports). As in our simulations, when applicable, both choices for $q$ led to very similar results, for simplicity, $q$ will be taken equal to $\pi$ as we continue.

The estimator of $u$ that we consider is then
\begin{equation}
  \label{eq:un}
  u^{[n]} = \Uc(p^{[n]})
\end{equation}
and, by analogy with classical copula modeling, could be called the \emph{empirical copula p.m.f.}

Denoting convergence in probability with the arrow $\p$, the consistency of $u^{[n]}$ can be immediately deduced from the consistency of $\hat p^{[n]}$ in~\eqref{eq:hat:pn}, the continuity of the $I$-projection of a strictly positive bivariate p.m.f.\ on a Fréchet class (see Lemma~\ref{lem:continuity} in Appendix~\ref{proof:prop:diff:I:proj}) and the continuous mapping theorem.

\begin{prop}[Consistency of $u^{[n]}$]
  \label{prop:consistency:un}
  Assume that $\hat p^{[n]} \p p$ in $\R^{r \times s}$, where $\hat p^{[n]}$ is defined in~\eqref{eq:hat:pn}. Then, $u^{[n]} \p u$ in $\R^{r \times s}$.
\end{prop}

The next proposition, proven in Appendix~\ref{proof:prop:asym:un}, gives the limiting distribution of $\sqrt{n} (u^{[n]} - u)$ in $\R^{r \times s}$. It is mostly a consequence of Proposition~\ref{prop:diff:I:proj} and the delta method \cite[see, e.g.,][Theorem 3.1]{van98}.

\begin{prop}[Limiting distribution of $\sqrt{n} (u^{[n]} - u)$]
  \label{prop:asym:un}
Assume that $\sqrt{n} (\hat p^{[n]} - p) \leadsto Z_p$ in $\R^{r \times s}$, where $\hat p^{[n]}$ is defined in~\eqref{eq:hat:pn}, the arrow $\leadsto$ denotes weak convergence and $Z_p$ is a random element of $\R^{r \times s}$. Then
$$
\sqrt{n} (u^{[n]} - u) = \Uc'_p( \sqrt{n} (\hat p^{[n]} - p) ) + o_P(1),
$$
where
\begin{itemize}
\item $\Uc_p'$ is the map from $\R^{r \times s}$ to $\R^{r \times s}$ defined by
\begin{equation*}
  \label{eq:Uc:prime}
  \Uc_p'(h)=\vect^{-1} ( J_{u,p} \vect(h) ), \qquad h \in \R^{r \times s},
\end{equation*}

\item $\vect$ is the operator defined as in~\eqref{eq:vect:S} with $S = I_{r,s}$,
\item $J_{u,p}$ is the $rs \times rs$ matrix given by
  \begin{equation}
    \label{eq:J:u:p}
    J_{u,p} = - K L_u^{-1} M_p N,
  \end{equation}
\item $K$ is the $rs \times (r-1)(s-1)$ matrix given by
\begin{equation*}
\begin{bmatrix}
  Q  & 0 & \dots & 0 \\
  0 & Q & \dots & 0 \\
  \vdots & \vdots & \ddots & \vdots \\
  0 & 0 & \dots & Q \\
  -Q & -Q & \dots & -Q \\
\end{bmatrix}
\text{ where }
Q =
 \begin{bmatrix}
   1  & 0 & \dots & 0 \\
   0 & 1 & \dots & 0 \\
   \vdots & \vdots & \ddots & \vdots \\
   0 & 0 & \dots & 1 \\
   -1 & -1 & \dots & -1 \\
 \end{bmatrix}
 \in \R^{r \times (r-1)},
\end{equation*}
\item $L_u$ is the  $(r-1)(s-1) \times (r-1)(s-1)$ matrix whose element at row $k + (r-1)(l-1)$, $(k,l) \in I_{r-1,s-1}$, and column $i + (r-1)(j-1)$, $(i,j) \in I_{r-1,s-1}$, is given by
\begin{equation}
  \label{eq:Lu}
   \frac{\1(i = k, j = l)}{u_{kl}}  +  \frac{\1(j = l)}{u_{rl}} + \frac{\1(i = k)}{u_{ks}} + \frac{1}{u_{rs}},
\end{equation}
\item $M_p$ is the  $(r-1)(s-1) \times (rs - 1)$ matrix whose element at row $k + (r-1)(l-1)$, $(k,l) \in I_{r-1,s-1}$, and column $i + r(j-1)$, $(i,j) \in I_{rs} \setminus \{(i,j)\}$, is given by
\begin{equation}
  \label{eq:Mp}
 - \frac{\1(i = k, j = l)}{p_{kl}} + \frac{\1(i = r, j = l)}{p_{rl}} + \frac{\1(i = k, j = s)}{p_{ks}} + \frac{1}{p_{rs}}, 
\end{equation}
\item and $N$ is the  $(rs - 1) \times rs$ matrix given by
  $$
   \begin{bmatrix}
   1  & 0 & \dots & 0 & 0\\
   0 & 1 & \dots & 0 & 0\\
   \vdots & \vdots & \ddots & \vdots \\
   0 & 0 & \dots & 1 & 0\\
 \end{bmatrix}.
 $$
\end{itemize}

Consequently,
$$
\sqrt{n} (u^{[n]} - u) \leadsto \Uc'_p(Z_p) \text{ in } \R^{r \times s}.
$$
Moreover, when $(X_1,Y_1),\dots,(X_n,Y_n)$ are independent copies of $(X,Y)$, $\vect(Z_p)$ is a $rs$-dimensional centered normal random vector with covariance matrix $\Sigma_p = \diag(\vect(p))  - \vect(p) \vect(p)^\top$ and $\vect(\Uc'_p(Z_p))$ is a $rs$-dimensional centered normal random vector with covariance matrix $J_{u,p} \Sigma_p J_{u,p}^\top$. 
\end{prop}

\begin{remark}
\label{rem:min:div:est}
From~\eqref{eq:I:proj:Frechet} and~\eqref{eq:I:proj:unif}, the empirical copula p.m.f.\ $u^{[n]}$ defined in~\eqref{eq:un} can be rewritten more explicitly as
\begin{equation}
  \label{eq:un:explicit}
u^{[n]} = \arginf_{v \in \Gamma_{\text{unif}}} D(v \| p^{[n]}).
\end{equation}
As noted by a Referee, the above estimator is a \emph{minimum divergence estimator} \citep[see, e.g.,][]{ReaCre88, MorParVaj95, BasShiPar11}. Such estimators have been studied for families of divergences (such as power-divergences in \cite{ReaCre88} and $\phi$-divergences in \cite{MorParVaj95}) which contain the Kullback--Leibler divergence as a particular case.  When based on the Kullback--Leibler divergence and with our notation, such estimators can be expressed as
\begin{equation}
  \label{eq:min:div:est}
v^{[n]} = \arginf_{v \in \Pi_0} D(v \| p^{[n]}),
\end{equation}
where $\Pi_0$ is a subset of interest of $\Gamma$ in~\eqref{eq:Gamma}. The strong similarity between~\eqref{eq:un:explicit} and~\eqref{eq:min:div:est} suggests that asymptotic results for $u^{[n]}$ should follow from asymptotic results for $v^{[n]}$ upon taking $\Pi_0 =  \Gamma_{\text{unif}}$. An inspection of the appendices of \cite{ReaCre88} for instance reveals however that asymptotic results for $v^{[n]}$ are typically obtained under the \emph{null hypothesis} that $p \in \Pi_0$ (which is fully meaningful given that the aforementioned reference focuses on goodness-of-fit testing). In other words, upon taking $\Pi_0$ equal to $\Gamma_{\text{unif}}$, typical asymptotic results for minimum divergence estimators would allow us to obtain the asymptotics of $u^{[n]}$ under the assumption that $p \in \Gamma_{\text{unif}}$. The latter would however obviously not be of much interest since the approach studied in this work implicitly assumes that the unknown bivariate p.m.f.\ $p$ does not have uniform margins in general. The above connection with minimum divergence estimators reveals nonetheless that the results stated in Proposition~\ref{prop:asym:un} should be a particular case of asymptotic results for minimum divergence estimators under the \emph{alternative hypothesis} that $p \not\in \Pi_0$, that is, \emph{under misspecification}. After a literature review, we were only able to find one reference that addresses this issue: it is the work of \cite{JimPinAlbMor11}. Specifically, in principle, Theorem~1 in the aforementioned reference could be used as an important building block to obtain an alternative proof of Proposition~\ref{prop:asym:un} above. However, its proof seems incomplete as its relies on a lemma whose proof seems incomplete as already mentioned in Remark~\ref{rem:JimEtAl}. \qed
\end{remark}

Recall the definitions of the moments $\rho$, $\gamma$ and $\tau$ of $p$ (or $u$) considered at the end of Section~\ref{sec:sklar:like} in~\eqref{eq:Yule} and~\eqref{eq:gamma:tau}. Natural estimators of the latter then simply follow from the plugin principle and are
\begin{equation}
\label{eq:rho:gamma:tau:estimators}
\rho^{[n]} = \Upsilon(u^{[n]}), \qquad \gamma^{[n]} = G(u^{[n]}) \qquad \text{and} \qquad \tau^{[n]} = T(u^{[n]}),
\end{equation}
respectively, where $u^{[n]}$ is defined in~\eqref{eq:un} and the maps $\Upsilon$, $G$ and $T$ are defined in~\eqref{eq:Upsilon}, \eqref{eq:G} and~\eqref{eq:T}, respectively.

As we continue, for an arbitrary map $\eta:\Gamma \to \R$, we define its gradient $\dot \eta$ to be the usual gradient written with respect to the standard column-major vectorization of its $r \times s$ matrix argument, that is,
$$
\dot \eta = \left( \frac{\partial\eta}{\partial x_{11}}, \dots, \frac{\partial\eta}{\partial x_{r1}}, \dots, \frac{\partial\eta}{\partial x_{1s}}, \dots, \frac{\partial\eta}{\partial x_{rs}} \right).
$$
The following result is then an immediate corollary of Propositions~\ref{prop:consistency:un} and~\ref{prop:asym:un}, the continuous mapping theorem and the delta method.

\begin{cor}[Asymptotics of moment estimators]
  \label{cor:asym:moment:est}
  If $\hat p^{[n]} \p p$ in $\R^{r \times s}$, where $\hat p^{[n]}$ is defined in~\eqref{eq:hat:pn}, then $\rho^{[n]} \p \rho$, $\gamma^{[n]} \p \gamma$ and $\tau^{[n]} \p \tau$ in $\R$. If, additionally, $\sqrt{n} (\hat p^{[n]} - p)$ converges weakly in $\R^{r \times s}$, then
\begin{align*}
  \sqrt{n} (\rho^{[n]} - \rho) &= \dot \Upsilon(u)^\top  J_{u,p} \sqrt{n} \, \vect(\hat p^{[n]} - p) + o_P(1), \\
  \sqrt{n} (\gamma^{[n]} - \gamma) &=  \dot G(u)^\top  J_{u,p} \sqrt{n} \, \vect(\hat p^{[n]} - p) ) + o_P(1), \\
  \sqrt{n} (\tau^{[n]} - \tau) &=  \dot T(u)^\top  J_{u,p} \sqrt{n} \, \vect(\hat p^{[n]} - p) ) + o_P(1),
\end{align*}
where $J_{u,p}$ is defined in~\eqref{eq:J:u:p}. Consequently, when $(X_1,Y_1)$, \dots, $(X_n,Y_n)$ are independent copies of $(X,Y)$, the sequences $\sqrt{n} (\rho^{[n]} - \rho)$, $\sqrt{n} (\gamma^{[n]} - \gamma)$ and $\sqrt{n} (\tau^{[n]} - \tau)$) are asymptotically centered normal with variances
\begin{gather*}
\dot \Upsilon(u)^\top J_{u,p} \Sigma_p J_{u,p}^\top  \dot \Upsilon(u),  \\
\dot G(u)^\top J_{u,p} \Sigma_p J_{u,p}^\top  \dot G(u), \\
\dot T(u)^\top J_{u,p} \Sigma_p J_{u,p}^\top  \dot T(u),
\end{gather*}
respectively, where $\Sigma_p = \diag(\vect(p))  - \vect(p) \vect(p)^\top$. 
\end{cor}

\subsection{Parametric estimation}
\label{sec:par:est}

Let
\begin{equation}
  \label{eq:Jc}
  \Jc = \{u^{[\theta]} : \theta \in \Theta\}
\end{equation}
be a bivariate parametric family of copula p.m.f.s on $I_{r,s}$, where $\Theta$ is an open subset of $\R^m$ for some strictly positive integer $m$. Because the assumption $\supp{(p)} = I_{r,s}$ implies that $\supp{(u)} = I_{r,s}$, the family $\Jc$ is naturally assumed to satisfy the following: for any $\theta \in \Theta$, $u^{[\theta]}_{ij} > 0$ for all $(i,j) \in I_{r,s}$. In this section, we assume that the unknown copula p.m.f.\ $u$ in~\eqref{eq:sklar:like} belongs to $\Jc$, that is, there exists $\theta_0 \in \Theta$ such that $u = u^{[\theta_0]}$, and our goal is to address the estimation of~$\theta_0$.

Before considering two estimation approaches, note that several examples of parametric copula p.m.f.s can be found in Section~7 of \cite{Gee20}. Under the assumption of rectangular support for $p$ (and thus $u$), it is particularly meaningful to follow one of the suggestions therein and construct the family $\Jc$ from a parametric family $\{C_\theta : \theta \in \Theta \}$ of classical bivariate copulas with strictly positive densities on $(0,1)^2$ such that, for any $\theta \in \Theta$ and $(i,j) \in I_{r,s}$,
\begin{equation}
  \label{eq:u:theta}
u^{[\theta]}_{ij} = C_\theta \left( \frac{i}{r}, \frac{j}{s} \right) - C_\theta \left( \frac{i}{r}, \frac{j-1}{s} \right) - C_\theta \left( \frac{i-1}{r}, \frac{j}{s} \right) + C_\theta \left( \frac{i-1}{r}, \frac{j-1}{s} \right).
\end{equation}
Of course, this way of proceeding is fully meaningful only if the resulting family $\Jc$ in~\eqref{eq:Jc} is identifiable, that is, $u^{[\theta]} \neq u^{[\theta']}$ whenever $\theta \neq \theta'$. To check that a family $\Jc$ is identifiable, it thus suffices to verify that, for any $\theta \neq \theta'$, there exists $(i,j) \in I_{r,s}$ such that $u^{[\theta]}_{ij} \neq u^{[\theta']}_{ij}$. Note that the quantity $u^{[\theta]}_{ij}$ in~\eqref{eq:u:theta} is actually the $C_\theta$-volume of the rectangle $((i-1)/r, i/r] \times ((j-1)/s, j/s]$  \citep[see, e.g.,][Section 2.1]{HofKojMaeYan18}. Non-identifiability thus occurs when a change in $\theta$ leaves the $C_\theta$-volumes of all the rectangles $((i-1)/r, i/r] \times ((j-1)/s, j/s]$, $(i,j) \in I_{r,s}$, unchanged. Since the construction in~\eqref{eq:u:theta} is based on classical copula families (which are identifiable), arguably, non-identifiability of the resulting family $\Jc$ in~\eqref{eq:Jc} should be the exception rather than the rule in particular as $r$ and $s$ get larger. For a parametric copula family for which $C_\theta$ has an explicit expression, identifiability of the family $\Jc$ can be checked analytically by replacing $C_\theta$ by its expression in~\eqref{eq:u:theta}. The latter is for instance done in \citet[Section 7.1]{Gee20} for the Farlie--Gumbel--Morgenstern family. For parametric copula families for which $C_\theta$ is not explicitly available (such as elliptical copula families),  identifiability could be checked numerically.

\subsubsection{Method-of-moments estimation}

We shall assume in this subsection that $\Jc$ is a one-parameter family, that is, $m=1$. Given~$\Jc$, let $g_\rho$, $g_\gamma$ and $g_\tau$ be the functions defined, for any $\theta \in \Theta$, by
\begin{equation}
  \label{eq:g:func}
g_\rho(\theta) = \Upsilon(u^{[\theta]}), \qquad g_\gamma(\theta) = G(u^{[\theta]}) \qquad \text{and} \qquad g_\tau(\theta) = T(u^{[\theta]}),
\end{equation}
where the maps $\Upsilon$, $G$ and $T$ are defined in~\eqref{eq:Upsilon}, \eqref{eq:G} and~\eqref{eq:T}, respectively. Method-of-moments estimators based on Yule's coefficient, the gamma coefficient or the tau coefficient can be used if the functions $g_\rho$, $g_\gamma$ and $g_\tau$ are one-to-one. In that case, corresponding estimators of $\theta_0$ are simply given by
\begin{equation}
  \label{eq:MoM:estimators}
\theta_{\rho}^{[n]} = g_\rho^{-1}(\rho^{[n]}), \qquad \theta_{\gamma}^{[n]} = g_\gamma^{-1}(\gamma^{[n]}) \qquad \text{and} \qquad \theta_{\tau}^{[n]} = g_\tau^{-1}(\tau^{[n]})
\end{equation}
where $\rho^{[n]}$, $\gamma^{[n]}$ and $\tau^{[n]}$ are the estimators of $\rho$, $\gamma$ and $\tau$, respectively, defined in~\eqref{eq:rho:gamma:tau:estimators}.

The following result is then an immediate consequence of Corollary~\ref{cor:asym:moment:est}, the continuous mapping theorem and the delta method.

\begin{cor}[Asymptotics of method-of-moments estimators]
  \label{cor:MoM}
Assume that the functions $g_\rho$, $g_\gamma$ and $g_\tau$ in~\eqref{eq:g:func} are one-to-one and that $g_\rho^{-1}$, $g_\gamma^{-1}$ and $g_\tau^{-1}$ are continuously differentiable at $\rho_0 = \Upsilon(u^{[\theta_0]})$, $\gamma_0 = G(u^{[\theta_0]})$ and $\tau_0 = T(u^{[\theta_0]})$. If $\hat p^{[n]} \p p$ in $\R^{r \times s}$, then $\theta_{\rho}^{[n]} \p \theta_0$,  $\theta_{\gamma}^{[n]} \p \theta_0$ and  $\theta_{\tau}^{[n]} \p \theta_0$ in $\R$. If, additionally, $\sqrt{n} (\hat p^{[n]} - p)$ converges weakly in $\R^{r \times s}$, then
\begin{align*}
  \sqrt{n} (\theta_{\rho}^{[n]} - \theta_0) &= \frac{\dot \Upsilon(u)^\top  J_{u,p} \sqrt{n} \, \vect(\hat p^{[n]} - p)}{g'_\rho(\theta_0)}  + o_P(1), \\
  \sqrt{n} (\theta_{\gamma}^{[n]} - \theta_0) &=  \frac{\dot G(u)^\top  J_{u,p} \sqrt{n} \, \vect(\hat p^{[n]} - p)}{g'_\gamma(\theta_0)} + o_P(1), \\
  \sqrt{n} ( \theta_{\tau}^{[n]} - \theta_0) &=  \frac{\dot T(u)^\top  J_{u,p} \sqrt{n} \, \vect(\hat p^{[n]} - p)}{g'_\tau(\theta_0)} + o_P(1),
\end{align*}
where $J_{u,p}$ is defined in~\eqref{eq:J:u:p}.
\end{cor}

\subsubsection{Maximum pseudo-likelihood estimation}
\label{sec:MPL}

We assume in this subsection that $\Jc$ in~\eqref{eq:Jc} is a multi-parameter family, that is, $m \geq 1$, with $\Theta$ an open subset of $\R^m$. Recall that we work under the assumption that there exists $\theta_0 \in \Theta$ such that $u = u^{[\theta_0]}$ and that our goal is to estimate the unknown parameter vector $\theta_0$. It is important to note that we do not have at our disposal observed counts from the bivariate p.m.f.\ $u = u^{[\theta_0]}$. We instead only have access to the observed counts $n \hat p^{[n]}$ from~$p$, where $\hat p^{[n]}$ is defined in~\eqref{eq:hat:pn}. To carry out maximum likelihood estimation, we would thus additionally need to postulate marginal parametric models for $p^{[1]}$ and $p^{[2]}$ and obtain an estimate of $\theta_0$ as a by-product of the estimation of all the parameters of a model for $p$ (see Remark~\ref{rem:MLE} below). Instead of considering such an intricate way of proceeding, it is conceptually simpler to consider minimum divergence estimators of the form $\check \theta^{[n]} = \arginf_{\theta \in \Theta} D(u^{[\theta]} \| u^{[n]})$
or
\begin{equation}
  \label{eq:mde}
  \theta^{[n]} = \arginf_{\theta \in \Theta} D(u^{[n]} \| u^{[\theta]}) =  \argsup_{\theta \in \Theta} \sum_{(i,j) \in I_{r,s}} u^{[n]}_{ij} \log u^{[\theta]}_{ij},
\end{equation}
where $D$ is the Kullback--Leibler divergence defined in~\eqref{eq:KL:divergence}. Notice that if the numbers in $n u^{[n]}$ were counts obtained from a random sample from $u = u^{[\theta_0]}$,
\begin{equation}
  \label{eq:PL}
\bar L^{[n]}(\theta) = n \sum_{(i,j) \in I_{r,s}} u^{[n]}_{ij} \log u^{[\theta]}_{ij}, \qquad \theta \in \Theta,
\end{equation}
would be the log-likelihood of the model. As the numbers in $n u^{[n]}$ are only a proxy to observed counts from $u^{[\theta_0]}$, the estimator in~\eqref{eq:mde}, which can be rewritten as
\begin{equation}
  \label{eq:theta:n}
\theta^{[n]} = (\theta^{[n]}_1,\dots,\theta^{[n]}_m) = \argsup_{\theta \in \Theta} \bar L^{[n]}(\theta) = \argsup_{\theta \in \Theta} \frac{1}{n} \bar L^{[n]}(\theta),
\end{equation}
is a maximum \emph{pseudo-likelihood} estimator of $\theta_0$. The aim of this section is to derive its consistency and its asymptotic normality.

\begin{remark}[Connection to minimum divergence estimators in multinomial models] Let $\Fc = \{p^{[\delta]} : \delta \in \Delta \}$ be a parametric family of bivariate p.m.f.s, where $\Delta$ is a open subset of $\R^d$ for some strictly positive integer $d$. Assume furthermore that $(X_1,Y_1),\dots,(X_n,Y_n)$ is a random sample from $p$ (which implies that $n \vect(\hat p^{[n]})$ is a multinomial random vector with parameters $n$ and $\vect(p)$, where $\hat p^{[n]}$ is defined in~\eqref{eq:hat:pn} and $\vect$ is the operator defined as in~\eqref{eq:vect:S} with $S = I_{r,s}$) and that there exists a $\delta_0 \in \Delta$ such that $p = p^{[\delta_0]}$. From~\eqref{eq:mde}, we see that the maximum pseudo-likelihood estimator in~\eqref{eq:theta:n} bears a strong resemblance with the estimator $\delta^{[n]} = \arginf_{\delta \in \Delta} D(p^{[n]} \| p^{[\delta]})$. The latter belongs to the classes of minimum divergence estimators of $\delta_0$ studied for instance in \cite{ReaCre88}, \cite{MorParVaj95} or \cite{BasShiPar11}. Because the numbers in $n u^{[n]}$ are not observed counts from $u^{[\theta_0]}$, the consistency and the asymptotic normality of~\eqref{eq:theta:n} cannot unfortunately be directly deduced from the asymptotic results stated in the aforementioned references. \qed
\end{remark}

\begin{remark}[Connection to maximum likelihood estimators]
  \label{rem:MLE}
Let $\Mc_1 = \{p^{[1,\alpha]} : \alpha \in A\}$ (resp.\ $\Mc_2 = \{p^{[2,\beta]} : \beta \in B\}$) be a univariate parametric family of p.m.f.s on $I_r$ (resp.\ $I_s$), where $A$ (resp.\ $B$) is an open subset of $\R^{m_1}$ (resp.\ $\R^{m_2}$) for some strictly positive integer $m_1$ (resp.\ $m_2$). The families $\Mc_1$ and $\Mc_2$ are further naturally assumed to satisfy the following: for any $(\alpha,\beta) \in A \times B$, $p^{[1,\alpha]}_i > 0$ for all $i \in I_r$ and $p^{[2,\beta]}_j > 0$ for all $j \in I_s$. Having~\eqref{eq:sklar:like} in mind, one can combine the previous marginal parametric assumptions with~\eqref{eq:Jc} to form a parametric model for $p$ as
\begin{equation}
  \label{eq:model}
  \Pc = \{p^{[\alpha,\beta,\theta]} = \Ic_{p^{[1,\alpha]}, p^{[2,\beta]}}(u^{[\theta]}) : (\alpha,\beta,\theta) \in A \times B \times \Theta \}.
\end{equation}
Under the assumption that there exists $(\alpha_0,\beta_0,\theta_0) \in (A,B,\Theta)$ such that $p = p^{[\alpha_0,\beta_0,\theta_0]}$, the estimation of $\theta_0$ is then a by-product of the estimation of the entire parameter vector $(\alpha_0,\beta_0,\theta_0)$. When $(X_1,Y_1),\dots,(X_n,Y_n)$ is a random sample from $p$, the latter could be obtained by maximizing the log-likelihood of the model in~\eqref{eq:model}, which can be expressed as
$$
L^{[n]}(\alpha,\beta,\theta) = n \sum_{(i,j) \in I_{r,s}} \hat p^{[n]}_{ij} \log p^{[\alpha,\beta,\theta]}_{ij} = n \sum_{(i,j) \in I_{r,s}} \hat p^{[n]}_{ij} \log \Ic_{p^{[1,\alpha]}, p^{[2,\beta]}}(u^{[\theta]})_{ij},
$$
where $\hat p_n$ is defined in~\eqref{eq:hat:pn}. We can expect that, in practice, the previous optimization would be computationally costly because it would typically require many executions of the IPFP. Furthermore, as when indirectly estimating the parameter vector of a classical parametric copula by maximum likelihood estimation \citep[see, e.g.,][Chapter~4 and the references therein]{HofKojMaeYan18}, the resulting estimate of $\theta_0$ would be affected by potential misspecification of the univariate families $\Mc_1$ and $\Mc_2$. A less computationally expensive ``two-stage'' approach would consist of first estimating $\alpha_0$ (resp.\ $\beta_0$) by $\alpha^{[n]}$ (resp.\ $\beta^{[n]}$) from the first (resp.\ second) component sample $X_1,\dots,X_n$ (resp.\ $Y_1,\dots,Y_n$) and then maximizing the following log-pseudo-likelihood:
\begin{align*}
  \tilde L^{[n]}(\theta) &= n \sum_{(i,j) \in I_{r,s}} \hat p^{[n]}_{ij} \log p^{[\alpha^{[n]},\beta^{[n]},\theta]}_{ij} \\
                         &= n \sum_{(i,j) \in I_{r,s}} \hat p^{[n]}_{ij} \log  \Ic_{p^{[1,\alpha^{[n]}]}, p^{[2,\beta^{[n]}]}}(u^{[\theta]})_{ij}.
\end{align*}
From a computational perspective, we can expect that the maximization of $\tilde L^{[n]}$ would be substantially more costly than the one of $\bar L^{[n]}$ in~\eqref{eq:PL} because the former would typically require many executions of the IPFP. Furthermore, as previously, the resulting estimate of $\theta_0$ would be affected by potential misspecification of the univariate families $\Mc_1$ and $\Mc_2$. This provides additional arguments in favor of the maximization of $\bar L^{[n]}$ in~\eqref{eq:PL} which, in spirit, is the analog of the log-pseudo-likelihood of \cite{GenGhoRiv95} in the current discrete context. \qed
\end{remark}

The following result, proven in Appendix~\ref{proof:prop:MPL} using Theorem~5.7 of \cite{van98}, provides conditions under which the estimator $\theta^{[n]}$ in~\eqref{eq:theta:n} is consistent.

\begin{prop}[Consistency of the maximum pseudo-likelihood estimator]
  \label{prop:MPL:consistency}
  Assume that the family $\Jc$ is identifiable and that there exists $\lambda \in (0,1)$ such that, for any $\theta \in \Theta$ and $(i,j) \in I_{r,s}$, $u^{[\theta]}_{ij}  \geq \lambda$. Then, if $\hat p^{[n]} \p p$ in $\R^{r \times s}$, where $\hat p^{[n]}$ is defined in~\eqref{eq:hat:pn}, $\theta^{[n]} \p \theta_0$ in $\R^m$.
\end{prop}

\begin{remark}
  The requirement that there exists $\lambda \in (0,1)$ such that, for any $\theta \in \Theta$ and $(i,j) \in I_{r,s}$, $u^{[\theta]}_{ij}  \geq \lambda$ might appear too restrictive. For instance, some thought reveals that it will not be satisfied by families $\Jc$ in~\eqref{eq:Jc} constructed via~\eqref{eq:u:theta} when $C_\theta$ is a Gumbel--Hougaard copula \citep[see, e.g.,][Chapter~2 and the references therein]{HofKojMaeYan18} if the parameter space $\Theta$ is taken equal to $(1,\infty)$. It will however be satisfied if the parameter space is restricted to $(1,M)$, for any fixed large real $M$. \qed
\end{remark}

To state conditions under which the estimator $\theta^{[n]}$ in~\eqref{eq:theta:n} is asymptotically normal, we need to introduce additional notation. For any $\theta \in \Theta$ and $(i,j) \in I_{r,s}$, let $\ell^{[\theta]}_{ij} = \log u^{[\theta]}_{ij}$, let
\begin{equation}
  \label{eq:dot:u:ij}
  \dot u^{[\theta]}_{ij, k} = \frac{\partial u^{[\theta]}_{ij}}{\partial \theta_k}, \,k \in I_m = \{1,\dots,m\},  \qquad \qquad \dot u^{[\theta]}_{ij} = \left(\dot u^{[\theta]}_{ij, 1}, \dots, \dot u^{[\theta]}_{ij, m} \right),
\end{equation}
and let
\begin{equation}
  \label{eq:ddot:u:ij}
  \ddot u^{[\theta]}_{ij, kl} = \frac{\partial^2 u^{[\theta]}_{ij}}{\partial \theta_k \partial \theta_l}, \,k,l \in I_m,  \qquad \qquad \ddot u^{[\theta]}_{ij} =
  \begin{bmatrix}
    \ddot u^{[\theta]}_{ij, 11} & \dots & \ddot u^{[\theta]}_{ij, 1m} \\
    \vdots & & \vdots \\
    \ddot u^{[\theta]}_{ij, m1} & \dots & \ddot u^{[\theta]}_{ij, mm}
    \end{bmatrix}.
\end{equation}
Similarly, for any $\theta \in \Theta$ and $(i,j) \in I_{r,s}$, let
\begin{equation}
  \label{eq:dot:l:ij}
  \dot \ell^{[\theta]}_{ij,k} = \frac{\partial \log u^{[\theta]}_{ij}}{\partial \theta_k} = \frac{\dot u^{[\theta]}_{ij, k}}{u^{[\theta]}_{ij}}, \, k \in I_m, \quad \qquad \dot \ell^{[\theta]}_{ij} = \left(\dot \ell^{[\theta]}_{ij,1},\dots, \dot \ell^{[\theta]}_{ij,m} \right) = \frac{\dot u^{[\theta]}_{ij}}{u^{[\theta]}_{ij}},
\end{equation}
let
\begin{equation}
  \label{eq:ddot:l:ij}
  \ddot \ell^{[\theta]}_{ij, kl} = \frac{\partial^2 \ell^{[\theta]}_{ij}}{\partial \theta_k \partial \theta_l} = \frac{\ddot u^{[\theta]}_{ij, kl}}{u^{[\theta]}_{ij}} - \frac{\dot u^{[\theta]}_{ij, k} \dot u^{[\theta]}_{ij, l}}{(u^{[\theta]}_{ij})^2}, \qquad k,l \in I_m,
\end{equation}
and let
\begin{equation}
  \label{eq:ddot:l:ij:matrix}
  \qquad \qquad \ddot \ell^{[\theta]}_{ij} =
  \begin{bmatrix}
    \ddot \ell^{[\theta]}_{ij, 11} & \dots & \ddot \ell^{[\theta]}_{ij, 1m} \\
    \vdots & & \vdots \\
    \ddot \ell^{[\theta]}_{ij, m1} & \dots & \ddot \ell^{[\theta]}_{ij, mm}
  \end{bmatrix}
  =  \frac{\ddot u^{[\theta]}_{ij}}{u^{[\theta]}_{ij}} - \frac{\dot u^{[\theta]}_{ij} \dot u^{[\theta],\top}_{ij}}{(u^{[\theta]}_{ij})^2},
\end{equation}
where $\dot u^{[\theta]}_{ij}$ and $\ddot u^{[\theta]}_{ij}$ are defined in~\eqref{eq:dot:u:ij} and~\eqref{eq:ddot:u:ij}, respectively. Using the fact that $\sum_{(i,j) \in I_{r,s}} u^{[\theta]}_{ij} = 1$ implies that $\sum_{(i,j) \in I_{r,s}} \dot u^{[\theta]}_{ij,k} = 0$ and $\sum_{(i,j) \in I_{r,s}} \ddot u^{[\theta]}_{ij,kl} = 0$ for all $k,l \in I_m$, we obtain from~\eqref{eq:dot:l:ij} and~\eqref{eq:ddot:l:ij}  that
$$
\sum_{(i,j) \in I_{r,s}} u^{[\theta]}_{ij} \dot \ell^{[\theta]}_{ij,k} = 0 \text{ for all } k \in I_m,
$$
and that
$$
\sum_{(i,j) \in I_{r,s}} u^{[\theta]}_{ij} \ddot \ell^{[\theta]}_{ij,kl} = - \sum_{(i,j) \in I_{r,s}} u^{[\theta]}_{ij} \dot \ell^{[\theta]}_{ij,k} \dot \ell^{[\theta]}_{ij,l} \text{ for all } k,l \in I_m.
$$
Using the notation defined in~\eqref{eq:dot:l:ij} and~\eqref{eq:ddot:l:ij:matrix}, the previous two centered displays can be rewritten as the vector identity $\Ex_\theta(\dot \ell^{[\theta]}_{(U,V)}) = 0$ and the matrix identity
\begin{equation}
  \label{eq:Fisher}
  \Ex_\theta(\ddot \ell^{[\theta]}_{(U,V)}) = - \Ex_\theta (\dot \ell^{[\theta]}_{(U,V)}\dot \ell^{[\theta],\top}_{(U,V)}),
\end{equation}
respectively, where $(U,V)$ has p.m.f.\ $u^{[\theta]}$ and $\Ex_\theta$ denotes the expectation with respect to $u^{[\theta]}$. In other words, in the discrete setting under consideration, unsurprisingly, we recover the classical identities that occur under regularity conditions in the context of classical maximum likelihood estimation \citep[see, e.g.,][Section 5.5, p 63]{van98}.

The following result is proven in Appendix~\ref{proof:prop:MPL}  along the lines of the proof of Theorem~5.21 in \cite{van98}.

\begin{prop}[Asymptotic normality of the maximum pseudo-likelihood estimator]
  \label{prop:MPL}
  Assume that $\theta^{[n]} \p \theta_0$ in $\R^m$ and, furthermore, that, for any $(i,j) \in I_{r,s}$, $\theta \mapsto \ell^{[\theta]}_{ij}$ is twice differentiable at any $\theta \in \Theta$ and that the matrix $\Ex_{\theta_0}(\ddot \ell^{[\theta_0]}_{(U,V)})$, where $(U,V)$ has p.m.f.\ $u^{[\theta_0]}$, is invertible. Then, if $\sqrt{n} (\hat p^{[n]} - p)$ converges weakly in $\R^{r \times s}$, we have that
$$
\sqrt{n} (\theta^{[n]} - \theta_0) = \{\Ex_{\theta_0} (\dot \ell^{[\theta_0]}_{(U,V)}\dot \ell^{[\theta_0],\top}_{(U,V)}) \}^{-1}  \dot \ell^{[\theta_0]}  J_{u,p} \sqrt{n} \, \vect(\hat p^{[n]} - p)  + o_P(1),
$$
where $\dot \ell^{[\theta_0]}$ is the $m \times rs$ matrix whose column $i + r(j-1)$, $(i,j) \in I_{r,s}$, $\dot \ell^{[\theta_0]}_{ij}$ is defined as in~\eqref{eq:dot:l:ij} with $\theta = \theta_0$ and $J_{u,p}$ is defined in~\eqref{eq:J:u:p}. Consequently, when $(X_1,Y_1),\dots,(X_n,Y_n)$ are independent copies of $(X,Y)$, the sequence $\sqrt{n} (\theta^{[n]} - \theta_0)$ is asymptotically centered normal with covariance matrix
$$
\{\Ex_{\theta_0} (\dot \ell^{[\theta_0]}_{(U,V)}\dot \ell^{[\theta_0],\top}_{(U,V)}) \}^{-1} \dot \ell^{[\theta_0]}  J_{u,p} \Sigma_p J_{u,p}^\top \dot \ell^{[\theta_0],\top} [\{ \Ex_{\theta_0} (\dot \ell^{[\theta_0]}_{(U,V)}\dot \ell^{[\theta_0],\top}_{(U,V)}) \}^{-1}]^\top,
$$
where $\Sigma_p = \diag(\vect(p))  - \vect(p) \vect(p)^\top$.
\end{prop}

\begin{remark}
The conditions on the hypothesized family $\Jc$ in~\eqref{eq:Jc} in the previous proposition are inspired by some of the weakest ones in the literature \citep[see, e.g.,][Chapter~5]{van98}. As such, we expect them to hold for many families $\Jc$. \qed
\end{remark}

\subsection{Monte Carlo experiments}
\label{sec:MC:est}

\begin{table}[t!]
\centering
\caption{For $(r,s) \in \{(3,3), (3,10), (10,10)\}$, bias and mean squared error (MSE) of the three method-of-moment estimators in~\eqref{eq:MoM:estimators} and the maximum pseudo-likelihood (PL) estimator in~\eqref{eq:theta:n} estimated from 1000 random samples of size $n \in \{100,500,1000\}$ generated, as explained in Section~\ref{sec:MC:est}, from p.m.f.s whose copula p.m.f.\ is of the form~\eqref{eq:u:theta} with $C_\theta$ the Clayton copula with a Kendall's tau of $0.33$. The column `m' gives the marginal scenario. The column `$\Uc$' reports the number of times the IPFP did not numerically converge in 1000 steps. The column `ni' report the number of numerical issues related to fitting.} 
\label{tab:fit:1:0.33}
\begingroup\scriptsize
\begin{tabular}{rrrrrrrrrrrrr}
  \hline
  \multicolumn{5}{c}{} & \multicolumn{4}{c}{Bias} & \multicolumn{4}{c}{MSE} \\ \cmidrule(lr){6-9} \cmidrule(lr){10-13} $(r,s)$ & m & n & $\Uc$ & ni & $\rho$ & $\gamma$ & $\tau$ & PL & $\rho$ & $\gamma$ & $\tau$ & PL \\ \hline
(3,3) & 1 & 100 & 0 & 0 & 0.06 & 0.05 & 0.06 & 0.06 & 0.13 & 0.13 & 0.13 & 0.12 \\ 
   &  & 500 & 0 & 0 & 0.01 & 0.01 & 0.01 & 0.01 & 0.02 & 0.02 & 0.02 & 0.02 \\ 
   &  & 1000 & 0 & 0 & 0.01 & 0.00 & 0.01 & 0.01 & 0.01 & 0.01 & 0.01 & 0.01 \\ 
   \\ & 2 & 100 & 0 & 0 & 0.06 & 0.04 & 0.05 & 0.06 & 0.16 & 0.15 & 0.16 & 0.16 \\ 
   &  & 500 & 0 & 0 & 0.01 & 0.00 & 0.01 & 0.01 & 0.03 & 0.03 & 0.03 & 0.03 \\ 
   &  & 1000 & 0 & 0 & 0.00 & 0.00 & 0.00 & 0.00 & 0.01 & 0.01 & 0.01 & 0.01 \\ 
   \\ & 3 & 100 & 0 & 0 & 0.08 & 0.06 & 0.07 & 0.07 & 0.20 & 0.19 & 0.19 & 0.18 \\ 
   &  & 500 & 0 & 0 & 0.01 & 0.01 & 0.01 & 0.01 & 0.03 & 0.03 & 0.03 & 0.03 \\ 
   &  & 1000 & 0 & 0 & 0.01 & 0.01 & 0.01 & 0.01 & 0.02 & 0.02 & 0.02 & 0.02 \\ 
   \\(3,10) & 1 & 100 & 0 & 0 & 0.05 & 0.03 & 0.05 & 0.05 & 0.11 & 0.11 & 0.11 & 0.10 \\ 
   &  & 500 & 0 & 0 & 0.01 & 0.00 & 0.01 & 0.01 & 0.02 & 0.02 & 0.02 & 0.02 \\ 
   &  & 1000 & 0 & 0 & 0.01 & 0.00 & 0.01 & 0.01 & 0.01 & 0.01 & 0.01 & 0.01 \\ 
   \\ & 2 & 100 & 0 & 0 & -0.07 & -0.10 & -0.08 & -0.09 & 0.17 & 0.17 & 0.17 & 0.22 \\ 
   &  & 500 & 0 & 0 & -0.01 & -0.01 & -0.01 & -0.00 & 0.03 & 0.03 & 0.03 & 0.04 \\ 
   &  & 1000 & 0 & 0 & 0.00 & -0.00 & 0.00 & 0.00 & 0.01 & 0.01 & 0.01 & 0.02 \\ 
   \\ & 3 & 100 & 0 & 0 & -0.58 & -0.58 & -0.57 & -0.64 & 0.45 & 0.45 & 0.45 & 0.51 \\ 
   &  & 500 & 0 & 0 & -0.26 & -0.27 & -0.26 & -0.29 & 0.24 & 0.24 & 0.24 & 0.29 \\ 
   &  & 1000 & 0 & 0 & -0.10 & -0.12 & -0.11 & -0.11 & 0.15 & 0.16 & 0.16 & 0.17 \\ 
   \\(10,10) & 1 & 100 & 0 & 0 & 0.04 & 0.01 & 0.03 & 0.03 & 0.09 & 0.08 & 0.09 & 0.07 \\ 
   &  & 500 & 0 & 0 & 0.01 & 0.00 & 0.01 & 0.01 & 0.01 & 0.01 & 0.01 & 0.01 \\ 
   &  & 1000 & 0 & 0 & 0.00 & 0.00 & 0.00 & 0.00 & 0.01 & 0.01 & 0.01 & 0.01 \\ 
   \\ & 2 & 100 & 0 & 1 & -0.12 & -0.14 & -0.12 & -0.13 & 0.14 & 0.14 & 0.14 & 0.19 \\ 
   &  & 500 & 0 & 0 & -0.02 & -0.02 & -0.02 & -0.02 & 0.03 & 0.03 & 0.03 & 0.04 \\ 
   &  & 1000 & 0 & 0 & -0.00 & -0.01 & -0.00 & -0.00 & 0.01 & 0.01 & 0.01 & 0.01 \\ 
   \\ & 3 & 100 & 0 & 1 & -0.86 & -0.83 & -0.83 & -0.82 & 0.76 & 0.72 & 0.71 & 0.68 \\ 
   &  & 500 & 0 & 0 & -0.68 & -0.65 & -0.64 & -0.66 & 0.51 & 0.48 & 0.47 & 0.50 \\ 
   &  & 1000 & 0 & 0 & -0.56 & -0.54 & -0.53 & -0.58 & 0.39 & 0.37 & 0.37 & 0.43 \\ 
   \hline
\end{tabular}
\endgroup
\end{table}

\begin{table}[t!]
\centering
\caption{For $(r,s) \in \{(3,3), (3,10), (10,10)\}$, bias and mean squared error (MSE) of the three method-of-moment estimators in~\eqref{eq:MoM:estimators} and the maximum pseudo-likelihood (PL) estimator in~\eqref{eq:theta:n} estimated from 1000 random samples of size $n \in \{100,500,1000\}$ generated, as explained in Section~\ref{sec:MC:est}, from p.m.f.s whose copula p.m.f.\ is of the form~\eqref{eq:u:theta} with $C_\theta$ the Clayton copula with a Kendall's tau of $0.66$. The column `m' gives the marginal scenario. The column `$\Uc$' reports the number of times the IPFP did not numerically converge in 1000 steps. The column `ni' report the number of numerical issues related to fitting.} 
\label{tab:fit:1:0.66}
\begingroup\scriptsize
\begin{tabular}{rrrrrrrrrrrrr}
  \hline
  \multicolumn{5}{c}{} & \multicolumn{4}{c}{Bias} & \multicolumn{4}{c}{MSE} \\ \cmidrule(lr){6-9} \cmidrule(lr){10-13} $(r,s)$ & m & n & $\Uc$ & ni & $\rho$ & $\gamma$ & $\tau$ & PL & $\rho$ & $\gamma$ & $\tau$ & PL \\ \hline
(3,3) & 1 & 100 & 0 & 0 & 0.12 & 0.04 & 0.12 & 0.12 & 0.85 & 0.77 & 0.88 & 0.81 \\ 
   &  & 500 & 0 & 0 & 0.02 & 0.01 & 0.02 & 0.02 & 0.13 & 0.14 & 0.13 & 0.13 \\ 
   &  & 1000 & 0 & 0 & 0.02 & 0.01 & 0.02 & 0.02 & 0.06 & 0.07 & 0.06 & 0.06 \\ 
   \\ & 2 & 100 & 0 & 0 & 0.12 & -0.00 & 0.13 & 0.12 & 0.93 & 0.73 & 0.96 & 0.90 \\ 
   &  & 500 & 0 & 0 & 0.02 & -0.00 & 0.02 & 0.02 & 0.14 & 0.14 & 0.14 & 0.15 \\ 
   &  & 1000 & 0 & 0 & 0.02 & 0.01 & 0.02 & 0.02 & 0.07 & 0.07 & 0.07 & 0.07 \\ 
   \\ & 3 & 100 & 0 & 0 & 0.07 & -0.00 & 0.07 & 0.07 & 0.80 & 0.77 & 0.82 & 0.76 \\ 
   &  & 500 & 0 & 0 & 0.02 & 0.01 & 0.02 & 0.02 & 0.15 & 0.16 & 0.14 & 0.15 \\ 
   &  & 1000 & 0 & 0 & 0.02 & 0.01 & 0.02 & 0.02 & 0.07 & 0.08 & 0.07 & 0.07 \\ 
   \\(3,10) & 1 & 100 & 0 & 0 & 0.03 & -0.04 & 0.04 & 0.03 & 0.58 & 0.58 & 0.61 & 0.54 \\ 
   &  & 500 & 0 & 0 & -0.00 & -0.02 & -0.00 & -0.00 & 0.10 & 0.10 & 0.10 & 0.09 \\ 
   &  & 1000 & 0 & 0 & 0.00 & -0.01 & -0.00 & 0.00 & 0.05 & 0.05 & 0.05 & 0.04 \\ 
   \\ & 2 & 100 & 0 & 0 & -0.49 & -0.52 & -0.44 & -0.64 & 0.94 & 0.96 & 0.92 & 1.28 \\ 
   &  & 500 & 0 & 0 & -0.04 & -0.05 & -0.03 & -0.05 & 0.11 & 0.11 & 0.11 & 0.12 \\ 
   &  & 1000 & 0 & 0 & -0.01 & -0.02 & -0.01 & -0.01 & 0.05 & 0.05 & 0.05 & 0.06 \\ 
   \\ & 3 & 100 & 0 & 0 & -2.63 & -2.59 & -2.55 & -2.83 & 7.27 & 7.09 & 6.89 & 8.41 \\ 
   &  & 500 & 0 & 0 & -1.44 & -1.40 & -1.36 & -1.59 & 3.06 & 3.00 & 2.87 & 3.62 \\ 
   &  & 1000 & 0 & 0 & -0.68 & -0.65 & -0.61 & -0.80 & 1.26 & 1.28 & 1.22 & 1.46 \\ 
   \\(10,10) & 1 & 100 & 19 & 0 & -0.10 & -0.16 & -0.04 & -0.04 & 0.40 & 0.39 & 0.39 & 0.31 \\ 
   &  & 500 & 0 & 0 & -0.01 & -0.03 & -0.01 & -0.00 & 0.07 & 0.07 & 0.07 & 0.05 \\ 
   &  & 1000 & 0 & 0 & -0.00 & -0.01 & -0.00 & 0.00 & 0.04 & 0.03 & 0.03 & 0.03 \\ 
   \\ & 2 & 100 & 6 & 0 & -0.50 & -0.53 & -0.41 & -0.52 & 0.70 & 0.71 & 0.64 & 0.98 \\ 
   &  & 500 & 28 & 0 & -0.04 & -0.05 & -0.03 & -0.02 & 0.07 & 0.07 & 0.07 & 0.08 \\ 
   &  & 1000 & 2 & 0 & -0.01 & -0.02 & -0.00 & 0.00 & 0.04 & 0.03 & 0.03 & 0.04 \\ 
   \\ & 3 & 100 & 0 & 0 & -3.30 & -3.11 & -3.08 & -3.23 & 11.01 & 9.83 & 9.70 & 10.71 \\ 
   &  & 500 & 0 & 0 & -2.37 & -2.10 & -2.05 & -2.09 & 6.27 & 5.11 & 4.95 & 5.30 \\ 
   &  & 1000 & 0 & 0 & -1.72 & -1.50 & -1.44 & -1.40 & 3.61 & 2.87 & 2.73 & 2.73 \\ 
   \hline
\end{tabular}
\endgroup
\end{table}

The asymptotic results stated in Corollary~\ref{cor:MoM} and Proposition~\ref{prop:MPL} do not provide any information on the finite-sample behavior of the three method-of-moments estimators in~\eqref{eq:MoM:estimators} and the maximum pseudo-likelihood estimator in~\eqref{eq:theta:n}. To compare these four estimators, we carried out Monte Carlo experiments in the situation when the parametric family $\Jc$ in~\eqref{eq:Jc} is constructed from a one-parameter family of classical copulas as in~\eqref{eq:u:theta}. For $(r,s) \in \{(3,3), (3,10), (10,10)\}$, the bias and the mean squared error (MSE) of the three method-of-moments estimators and the maximum pseudo-likelihood estimator were estimated from 1000 random samples of size $n \in \{100,500,1000\}$ generated from a p.m.f.\ with copula p.m.f.\ of the form~\eqref{eq:u:theta} with $C_\theta$ either the Clayton or the Gumbel--Hougaard copula with a Kendall's tau in $\{0.33, 0.66\}$, or the Frank copula with a Kendall's tau in $\{-0.5, 0, 0.5\}$ \citep[see, e.g.,][Chapter~2 for the definitions of these copula families and the definition of Kendall's tau]{HofKojMaeYan18}. For each of the above seven copula p.m.f.s, three marginal p.m.f.\ scenarios were considered:
\begin{enumerate}
\item $(1/r,\dots,1/r)$ (resp.\ $(1/s,\dots,1/s)$) as values of the first (resp.\ second) marginal p.m.f.,
\item $(1,\dots,r) / (r (r+1)/2)$ (resp.\ $(1,\dots,s) / (s (s+1)/2)$) as values of the first (resp.\ second) marginal p.m.f.,
\item the values of the p.m.f.\ of the binomial distribution with parameters $r-1$ and $1/2$ (resp.\ $s-1$ and $1/2$) as values of the first (resp.\ second) marginal p.m.f. 
\end{enumerate}
For the second and third marginal scenarios, the resulting p.m.f.s of the form~\eqref{eq:sklar:like} were computed using the IPFP. For each generated sample from one of the 21 data generating p.m.f.s, we first computed the nonparametric estimate $p^{[n]}$ in~\eqref{eq:pn:q}, then, using the IPFP, the corresponding empirical copula p.m.f.\ $u^{[n]}$ in~\eqref{eq:un} and, finally, the four estimates of $\theta_0$ using~\eqref{eq:MoM:estimators} and~\eqref{eq:theta:n}. All the computations were carried out using the \textsf{R} statistical environment \citep{Rsystem} and its packages \texttt{mipfp} \citep{mipfp} and \texttt{copula} \citep{copula}. In particular, the IPFP was computed using the function \texttt{Ipfp()} of the package \texttt{mipfp} with its default parameter values (at most 1000 iterations and $\eps$ in~\eqref{eq:IPFP:criterion} equal to $10^{-10}$), the inverses of the (one-to-one) functions $g_\rho$, $g_\gamma$ and $g_\tau$ in~\eqref{eq:g:func} were computed by numerical root finding using the \texttt{uniroot()} function while the maximization of the log-pseudo-likelihood in~\eqref{eq:PL} was carried out using the \texttt{optim()} function with $\theta^{[n]}_\gamma$ in~\eqref{eq:MoM:estimators} as starting value.

The results when the data generating p.m.f.\ is based on the copula p.m.f.~\eqref{eq:u:theta} with $C_\theta$ a Clayton copula are given in Tables~\ref{tab:fit:1:0.33} and~\ref{tab:fit:1:0.66}. Note that the column `$\Uc$' gives, for each data generating scenario, the number of times the IPFP did not numerically converge in 1000 steps in the sense of~\eqref{eq:IPFP:criterion}. The next column, `ni', reports the number of numerical issues related to fitting (either related to numerical root finding for the method-of-moments estimators or to numerical optimization for the maximum pseudo-likelihood). An examination of all the fitting simulation results revealed that, overall, the number of numerical issues was very small. The latter concern essentially the Clayton model with strongest dependence (see Table~\ref{tab:fit:1:0.66}) when $r = s = 10$, and could be explained by the higher probability of 0 counts in the bivariate contingency tables obtained from the generated samples in this case. In terms of estimation precision, as one can see, reassuringly, for each data generating scenario, the larger $n$, the smaller the absolute value of the bias and the MSE. Unsurprisingly, the worse results are obtained for the third marginal scenario since it is this scenario that leads to the largest number of very small probabilities for some of the cells of the data generating p.m.f. For instance, note the large negative biases when $r$ or $s$ are equal to $10$ as the smallest marginal probability is then approximately equal to 0.002 leading frequently to zero counts in several cells of the contingency tables obtained from the generated samples. Roughly speaking, in this case, the very non-uniform margins ``hide'' certain features of $u$ and make the dependence look substantially weaker on average than it actually is. From a numerical perspective, note also that, for such data generating scenarios, without the smoothing considered in~\eqref{eq:pn}, the IPFP would converge substantially less often in less than 1000 steps.

The results when the data generating p.m.f.\ is based on the copula p.m.f.~\eqref{eq:u:theta} with $C_\theta$ a Gumbel--Hougaard or a Frank copula are not qualitatively different and are not reported. In terms of MSE, the experiments did not reveal a uniformly better estimator.

\section{Goodness-of-fit testing}
\label{sec:GOF}

Recall the definition of the parametric family of copula p.m.f.s $\Jc$ in~\eqref{eq:Jc}. The three method-of-moments estimators in~\eqref{eq:MoM:estimators} and the maximum pseudo-likelihood estimator in~\eqref{eq:theta:n} were theoretically and empirically studied in Sections~\ref{sec:par:est} and~\ref{sec:MC:est} under the hypothesis
\begin{equation}
  \label{eq:H0}
  H_0 : u \in \Jc, \text{ that is, there exists } \theta_0 \in \Theta \text{ such that } u = u^{[\theta_0]}.
\end{equation}
Clearly, parameter estimates obtained for a given family $\Jc$ will be meaningful only if $H_0$ actually holds. It is the aim of goodness-of-fit testing to formally assess the latter hypothesis. To derive goodness-of-fit tests, we consider, as classically done in the literature, approaches based on comparing a nonparametric estimator of $u$ with a parametric one under $H_0$.

In the rest of this section, we first define a relevant class of chi-square statistics and provide their asymptotic null distributions. Next, we empirically study the finite-sample performance of the resulting asymptotic goodness-of-fit tests. Finally, we propose a semi-parametric bootstrap procedure that provides an alternative way of computing p-values.

\subsection{Asymptotic chi-square-type goodness-of-fit tests}
\label{sec:GOF:asym}

One approach to test $H_0:u \in \Jc$ versus $H_1:u \not \in \Jc$ consists of constructing test statistics as norms of the \emph{goodness-of-fit process}
\begin{equation}
  \label{eq:GOF:process}
  \sqrt{n} (u^{[n]} - u^{[\hat \theta^{[n]}]}),
\end{equation}
where $u^{[n]}$ is defined in~\eqref{eq:un} and $\hat \theta^{[n]}$ is an estimator of $\theta_0$ computed under~$H_0$. As classically in the goodness-of-fit testing literature, the process in~\eqref{eq:GOF:process} compares a nonparametric estimator of $u$ which is consistent whether $H_0$ is true or not, with a parametric estimator which is only consistent under $H_0$. The rationale behind this construction is that any norm of~\eqref{eq:GOF:process} should be typically smaller under $H_0$ than it is under $H_1$.

The following result, proven in Appendix~\ref{proofs:GOF}, describes the asymptotic null behavior of the goodness-of-fit process in~\eqref{eq:GOF:process}.

\begin{prop}
  \label{prop:GOF}
Assume that $H_0$ in~\eqref{eq:H0} holds and that
\begin{enumerate}
\item $\sqrt{n} (\hat p^{[n]} - p)$ converges weakly in $\R^{r \times s}$, where $\hat p^{[n]}$ is defined in~\eqref{eq:hat:pn},
\item the map from $\Theta \subset \R^m$ to $\R^{r \times s}$ defined by $\theta \mapsto u^{[\theta]}$ is differentiable at~$\theta_0$ and let $\dot u^{[\theta_0]}$ be the $rs \times m$ matrix whose row $i + r(j-1)$, $(i,j) \in I_{r,s}$, is $\dot u^{[\theta]}_{ij}$ in~\eqref{eq:dot:u:ij},
\item there exists a $m \times rs$ matrix $V_{u,p}^{[\theta_0]}$ such that
  $$
  \sqrt{n} (\hat \theta^{[n]} - \theta_0) = V_{u,p}^{[\theta_0]} \sqrt{n} \, \vect(\hat p^{[n]} - p) + o_P(1),
  $$
  where $\vect$ is the operator defined as in~\eqref{eq:vect:S} with $S = I_{r,s}$.
\end{enumerate}
Then,
\begin{equation*}
  \sqrt{n} \, \vect(u^{[n]} - u^{[\hat \theta^{[n]}]}) =  (J_{u,p} - \dot u^{[\theta_0]} V_{u,p}^{[\theta_0]}) \sqrt{n} \, \vect(\hat p^{[n]} - p) + o_P(1),
\end{equation*}
where $J_{u,p}$ is defined in~\eqref{eq:J:u:p}.
\end{prop}

Note that Corollary~\ref{cor:MoM} and Proposition~\ref{prop:MPL} provide conditions under which the third assumption in the previous proposition holds for the three method-of-moments estimators in~\eqref{eq:MoM:estimators} and the maximum pseudo-likelihood estimator in~\eqref{eq:theta:n}, respectively.

One natural test statistic that can be constructed from the goodness-of-fit process in~\eqref{eq:GOF:process} is the chi-square statistic
\begin{equation}
  \label{eq:S:n}
  S^{[n]} = n \sum_{(i,j) \in I_{r,s}} \frac{(u^{[n]}_{ij} - u^{[\hat \theta^{[n]}]}_{ij})^2}{u^{[\hat \theta^{[n]}]}_{ij}}.
\end{equation}
As is well-known in the statistical literature, a widely accepted rule of thumb to protect the level of a classical chi-square test is to regroup low observed counts such that eventually all observed counts are above 5 \citep[see, e.g.,][Chapter~17]{van98}. To be able to perform similar groupings in our context, we shall consider the following generalization of $S^{[n]}$ in~\eqref{eq:S:n}:
\begin{equation}
  \label{eq:S:n:G}
  S_G^{[n]} = n \| \diag(G \, \vect(u^{[\hat \theta^{[n]}]}))^{-1/2} G \, \vect(u^{[n]} - u^{[\hat \theta^{[n]}]}) \|_2^2,
\end{equation}
where $G$ is a chosen ``grouping matrix'', that is, a $q \times rs$ matrix with $q \in \{1, \dots, rs\}$ whose elements are in $\{0,1\}$ with exactly $rs$ of them equal to 1 and a unique 1 per column. Some thought reveals that, in~\eqref{eq:S:n:G}, each row of $G$ sums one or more copula p.m.f.\ values. Clearly, $S_G^{[n]}$ coincides with $S^{[n]}$ in~\eqref{eq:S:n} when $G$ is the $rs \times rs$ identity matrix.

The following result, proven in Appendix~\ref{proofs:GOF}, gives the asymptotic null distribution of $S_G^{[n]}$ in~\eqref{eq:S:n:G} when the latter is computed from a random sample.

\begin{prop}
  \label{prop:S:n:G}
  Under the conditions of Proposition~\ref{prop:GOF}, when $(X_1,Y_1)$, \dots, $(X_n,Y_n)$ are independent copies of $(X,Y)$, $S_G^{[n]} \leadsto \Lc_G \text{ in } \R$, where $\Lc_G$ is distributed as $\sum_{k=1}^q \lambda_k Z_k^2$ for i.i.d.\ $N(0,1)$-distributed random variables $Z_1,\dots,Z_q$ and $\lambda_1,\dots,\lambda_q$ the eigenvalues of
  \begin{multline}
    \label{eq:cov}
  \Sigma_{G,u,p}^{[\theta_0]} = \diag(G \, \vect(u^{[\theta_0]}))^{-1/2} G (J_{u,p} - \dot u^{[\theta_0]} V_{u,p}^{[\theta_0]}) \\ \times \Sigma_p (J_{u,p} - \dot u^{[\theta_0]} V_{u,p}^{[\theta_0]})^\top G^\top \diag(G \, \vect(u^{[\theta_0]}))^{-1/2},
\end{multline}
with $J_{u,p}$ defined in~\eqref{eq:J:u:p} and $\Sigma_p = \diag(\vect(p))  - \vect(p) \vect(p)^\top$. Furthermore, provided that $(\theta,u',p') \mapsto \Sigma_{G,u',p'}^{[\theta]}$ is continuous at $(\theta_0, u, p)$, we have that $\Sigma_{G,u^{[n]},p^{[n]}}^{[\hat \theta^{[n]}]} \p \Sigma_{G,u,p}^{[\theta_0]}$ in $\R^{q \times q}$, where $p^{[n]}$ is defined in~\eqref{eq:pn}.
\end{prop}

Note that the eigenvalues $\lambda_1,\dots,\lambda_q$ of $\Sigma_{G,u,p}^{[\theta_0]}$ are not in general equal to 0 or 1 so that $\Lc_G$ does not have a chi-square distribution in general. The previous proposition however suggests that a goodness-of-fit test based on $S_G^{[n]}$ in~\eqref{eq:S:n:G} could be carried out in practice as follows:
\begin{enumerate}
\item For the hypothesized parametric family of copula p.m.f.s $\Jc$ in~\eqref{eq:Jc}, estimate $\theta_0$ by $\hat \theta^{[n]}$, where $\hat \theta^{[n]}$ is one the three method-of-moments estimators in~\eqref{eq:MoM:estimators} or the maximum pseudo-likelihood estimator in~\eqref{eq:theta:n}, and compute $S_G^{[n]}$ in~\eqref{eq:S:n:G}.
\item Compute the eigenvalues $\hat \lambda_1,\dots, \hat \lambda_q$ of $\Sigma_{G,u^{[n]},p^{[n]}}^{[\hat \theta^{[n]}]}$ which estimate the eigenvalues $\lambda_1,\dots,\lambda_q$ of $\Sigma_{G,u,p}^{[\theta_0]}$ in~\eqref{eq:cov}.
\item For some large integer $M$, compute $S_G^{[n],l} = \sum_{k=1}^q \hat \lambda_k Z_{kl}^2$, $l \in \{1,\dots,M\}$, for i.i.d.\ $N(0,1)$-distributed random variables $Z_{kl}$, $k \in \{1,\dots,q\}$, $l \in \{1,\dots,M\}$, and estimate the p-value of the test as
  $$
  \frac{1}{M} \sum_{l = 1}^M \1(S_G^{[n],l} \geq S_G^{[n]}).
  $$
\end{enumerate}

As a proof of concept, we shall focus on the case when $\hat \theta^{[n]}$ is $\theta^{[n]}_\rho$ in~\eqref{eq:MoM:estimators}, the estimator of $\theta_0$ based on the inversion of Yule's coefficient. From Corollary~\ref{cor:MoM}, we then know that, in this case,
$$
V_{u,p}^{[\theta_0]} = \frac{1}{g'_\rho(\theta_0)} \dot \Upsilon(u)^\top  J_{u,p}.
$$
Standard calculations show that
$$
\dot \Upsilon(u) = \frac{12}{\sqrt{(r+1)(s+1)(r-1)(s-1)}} \big(  (i-1)(j-1)  \big)_{(i,j) \in I_{r,s}}
$$
and that
$$
g'_\rho(\theta_0) = \frac{12}{\sqrt{(r+1)(s+1)(r-1)(s-1)}} \sum_{(i,j) \in I_{r,s}} (i-1)(j-1) \, \dot u^{[\theta_0]}_{ij}.
$$
The previous formulas can be used to obtain the expression of the covariance matrix $\Sigma_{G,u,p}^{[\theta_0]}$ in~\eqref{eq:cov} in terms of $\dot u^{[\theta_0]}$. When the hypothesized family $\Jc$ in~\eqref{eq:Jc} is defined from a parametric copula family as in~\eqref{eq:u:theta}, $\dot u^{[\theta_0]}$ can be obtained by differentiating~\eqref{eq:u:theta}. The latter takes the form of standard calculations for parametric copula families for which $C_\theta$ has an explicit expression. However, it can also be done for so-called implicit copula families such as the normal or the~$t$ using the expressions obtained in \cite{KojYan11}.

\subsection{Monte Carlo experiments}
\label{sec:MC:gof}

\begin{table}[t!]
\centering
\caption{For $(r,s) \in \{(3,3), (3,5), (5,5)\}$, rejection percentages of the goodness-of-fit test based on $S^{[n]}$ in~\eqref{eq:S:n} and Yule's coefficient computed from 1000 random samples of size $n \in \{100,500,1000\}$ generated, as explained in Section~\ref{sec:MC:gof}, from p.m.f.s whose copula p.m.f.\ is of the form~\eqref{eq:u:theta} with $C_\theta$ the Clayton copula with a Kendall's tau in $\{0.33, 0.66\}$. The column `m' gives the marginal scenario. The integer between parentheses is the number of numerical issues encountered out of 1000 executions.} 
\label{tab:gof:1}
\begingroup\footnotesize
\begin{tabular}{rrrrrrrrrr}
  \hline
  \multicolumn{4}{c}{} & \multicolumn{3}{c}{$\tau=0.33$} & \multicolumn{3}{c}{$\tau=0.66$} \\ \cmidrule(lr){5-7} \cmidrule(lr){8-10} $r$ & $s$ & m & n & \multicolumn{1}{c}{Cl} & \multicolumn{1}{c}{GH} & \multicolumn{1}{c}{F} & \multicolumn{1}{c}{Cl} & \multicolumn{1}{c}{GH} & \multicolumn{1}{c}{F}  \\ \hline
3 & 3 & 1 & 100 & 4.5 (0) & 23 (0) & 11.7 (0) & 1.5 (0) & 21.9 (0) & 6.9 (0) \\ 
   &  &  & 500 & 5.4 (0) & 93.5 (0) & 58.8 (0) & 3 (0) & 97 (0) & 76.3 (0) \\ 
   &  &  & 1000 & 4.8 (0) & 99.9 (0) & 89.1 (0) & 3.9 (0) & 100 (0) & 99.5 (0) \\ 
   \\ &  & 2 & 100 & 2.6 (0) & 18.6 (0) & 7.2 (0) & 1 (0) & 13.6 (0) & 4.6 (0) \\ 
   &  &  & 500 & 4.4 (0) & 84.3 (0) & 42.6 (0) & 2 (0) & 95.1 (0) & 67.2 (0) \\ 
   &  &  & 1000 & 4.3 (0) & 99.3 (0) & 80.8 (0) & 3.4 (0) & 100 (0) & 97.5 (0) \\ 
   \\ &  & 3 & 100 & 2.9 (0) & 20.3 (0) & 7.6 (0) & 0.7 (0) & 16 (0) & 5.2 (0) \\ 
   &  &  & 500 & 5.2 (0) & 88.7 (0) & 55.3 (0) & 1.7 (0) & 96.7 (0) & 67.1 (0) \\ 
   &  &  & 1000 & 5.6 (0) & 99.9 (0) & 84.8 (0) & 2.4 (0) & 100 (0) & 97.1 (0) \\ 
   \\3 & 5 & 1 & 100 & 2 (0) & 19.8 (0) & 6.2 (0) & 1.2 (0) & 7.9 (0) & 0.5 (0) \\ 
   &  &  & 500 & 4.6 (0) & 96.9 (0) & 70 (0) & 2.6 (0) & 95.5 (0) & 81.3 (0) \\ 
   &  &  & 1000 & 4.6 (0) & 100 (0) & 96.2 (0) & 3.1 (0) & 99.8 (0) & 99.1 (0) \\ 
   \\ &  & 2 & 100 & 0.6 (0) & 5.4 (0) & 2 (0) & 0.5 (0) & 2.5 (0) & 0 (0) \\ 
   &  &  & 500 & 4.9 (0) & 89.2 (0) & 45.5 (0) & 3.8 (0) & 90 (0) & 35.7 (0) \\ 
   &  &  & 1000 & 5 (0) & 99.9 (0) & 83.4 (0) & 3 (0) & 98.6 (0) & 89 (0) \\ 
   \\ &  & 3 & 100 & 0.2 (0) & 1.9 (0) & 0.5 (0) & 0.6 (0) & 0.6 (0) & 0 (0) \\ 
   &  &  & 500 & 2.8 (0) & 64.7 (0) & 26.2 (0) & 2.2 (0) & 63.2 (0) & 12.6 (0) \\ 
   &  &  & 1000 & 3.4 (0) & 97.4 (0) & 65.1 (0) & 3.2 (0) & 94.5 (0) & 73.7 (0) \\ 
   \\5 & 5 & 1 & 100 & 0.5 (0) & 7.1 (0) & 1.4 (0) & 1.8 (2) & 7.2 (0) & 0.5 (1) \\ 
   &  &  & 500 & 4.1 (0) & 97.9 (0) & 76.6 (0) & 2.9 (0) & 87.6 (0) & 55.3 (0) \\ 
   &  &  & 1000 & 5.1 (0) & 100 (0) & 99.6 (0) & 4.3 (0) & 98.9 (0) & 90.6 (0) \\ 
   \\ &  & 2 & 100 & 0.1 (0) & 0.9 (0) & 0.2 (0) & 0.7 (0) & 4.1 (1) & 0 (0) \\ 
   &  &  & 500 & 2 (0) & 87.1 (0) & 43.8 (0) & 5.2 (0) & 93.8 (0) & 49.2 (0) \\ 
   &  &  & 1000 & 4.4 (0) & 100 (0) & 87.2 (0) & 5.1 (0) & 100 (0) & 86.2 (0) \\ 
   \\ &  & 3 & 100 & 0 (0) & 0 (0) & 0 (0) & 0.2 (0) & 0 (0) & 0 (0) \\ 
   &  &  & 500 & 0.5 (0) & 23.8 (0) & 4.9 (0) & 1.2 (0) & 40.4 (0) & 7.1 (0) \\ 
   &  &  & 1000 & 1.1 (0) & 77.1 (0) & 29.4 (0) & 3.3 (0) & 84.5 (0) & 67.7 (0) \\ 
   \hline
\end{tabular}
\endgroup
\end{table}

\begin{table}[t!]
\centering
\caption{For $(r,s) \in \{(3,3), (3,5), (5,5)\}$, rejection percentages of the goodness-of-fit test based on $S^{[n]}$ in~\eqref{eq:S:n} and Yule's coefficient computed from 1000 random samples of size $n \in \{100,500,1000\}$ generated, as explained in Section~\ref{sec:MC:gof}, from p.m.f.s whose copula p.m.f.\ is of the form~\eqref{eq:u:theta} with $C_\theta$ the Gumbel--Hougaard copula with a Kendall's tau in $\{0.33, 0.66\}$. The column `m' gives the marginal scenario. The integer between parentheses is the number of numerical issues encountered out of 1000 executions.} 
\label{tab:gof:2}
\begingroup\footnotesize
\begin{tabular}{rrrrrrrrrr}
  \hline
  \multicolumn{4}{c}{} & \multicolumn{3}{c}{$\tau=0.33$} & \multicolumn{3}{c}{$\tau=0.66$} \\ \cmidrule(lr){5-7} \cmidrule(lr){8-10} $r$ & $s$ & m & n & \multicolumn{1}{c}{Cl} & \multicolumn{1}{c}{GH} & \multicolumn{1}{c}{F} & \multicolumn{1}{c}{Cl} & \multicolumn{1}{c}{GH} & \multicolumn{1}{c}{F}  \\ \hline
3 & 3 & 1 & 100 & 27.8 (0) & 3.8 (0) & 3.9 (0) & 29.5 (0) & 0.9 (0) & 2.1 (0) \\ 
   &  &  & 500 & 94.8 (0) & 4.1 (0) & 18 (0) & 99.2 (0) & 3.8 (0) & 18.5 (0) \\ 
   &  &  & 1000 & 99.8 (0) & 6.1 (0) & 32.7 (0) & 100 (0) & 3.4 (0) & 41.7 (0) \\ 
   \\ &  & 2 & 100 & 17.6 (0) & 2.4 (0) & 4.6 (0) & 20 (0) & 0.6 (0) & 1.8 (0) \\ 
   &  &  & 500 & 88.2 (0) & 3.7 (0) & 14.3 (0) & 97.4 (0) & 2.2 (0) & 13.7 (0) \\ 
   &  &  & 1000 & 99.5 (0) & 4.6 (0) & 28.6 (0) & 100 (0) & 3.6 (0) & 32.7 (0) \\ 
   \\ &  & 3 & 100 & 20 (0) & 2.6 (0) & 4.2 (0) & 23.1 (0) & 0.9 (0) & 1.9 (0) \\ 
   &  &  & 500 & 89 (0) & 4.2 (0) & 17.2 (0) & 98.2 (0) & 2.1 (0) & 9.5 (0) \\ 
   &  &  & 1000 & 99.8 (0) & 5 (0) & 31.7 (0) & 100 (0) & 2.9 (0) & 29.8 (0) \\ 
   \\3 & 5 & 1 & 100 & 21.9 (0) & 2.4 (0) & 3.4 (0) & 23.7 (0) & 0.4 (0) & 0.1 (0) \\ 
   &  &  & 500 & 97 (0) & 4.8 (0) & 24.1 (0) & 99.8 (0) & 2 (0) & 18.8 (0) \\ 
   &  &  & 1000 & 100 (0) & 4.3 (0) & 53.3 (0) & 100 (0) & 2.7 (0) & 48.6 (0) \\ 
   \\ &  & 2 & 100 & 9.8 (0) & 0.3 (0) & 0.9 (0) & 16.9 (0) & 0 (0) & 0.1 (0) \\ 
   &  &  & 500 & 89.9 (0) & 6.3 (0) & 15.8 (0) & 98.6 (0) & 1.8 (0) & 8.2 (0) \\ 
   &  &  & 1000 & 100 (0) & 5.9 (0) & 34.7 (0) & 100 (0) & 2.9 (0) & 28.7 (0) \\ 
   \\ &  & 3 & 100 & 4.4 (0) & 0.1 (0) & 0.1 (0) & 3.1 (0) & 0 (0) & 0 (0) \\ 
   &  &  & 500 & 68.5 (0) & 2 (0) & 8.8 (0) & 90.2 (0) & 0.8 (0) & 1.3 (0) \\ 
   &  &  & 1000 & 97.3 (0) & 3.9 (0) & 22.4 (0) & 99.8 (0) & 1.5 (0) & 13.7 (0) \\ 
   \\5 & 5 & 1 & 100 & 10.6 (0) & 0.3 (0) & 0.9 (0) & 32.7 (0) & 0.1 (0) & 0.2 (0) \\ 
   &  &  & 500 & 98.8 (0) & 4.5 (0) & 31.9 (0) & 99.6 (0) & 1.7 (0) & 29.2 (0) \\ 
   &  &  & 1000 & 100 (0) & 5.1 (0) & 72.1 (0) & 100 (0) & 1.9 (0) & 83.9 (0) \\ 
   \\ &  & 2 & 100 & 2 (0) & 0 (0) & 0 (0) & 15.2 (0) & 0.1 (0) & 0.1 (0) \\ 
   &  &  & 500 & 88.3 (0) & 2.9 (0) & 14.7 (0) & 98.1 (0) & 1 (0) & 14.4 (0) \\ 
   &  &  & 1000 & 100 (0) & 5.2 (0) & 42.2 (0) & 100 (0) & 3 (0) & 62.5 (0) \\ 
   \\ &  & 3 & 100 & 0.2 (0) & 0 (0) & 0 (0) & 0.2 (0) & 0 (0) & 0 (0) \\ 
   &  &  & 500 & 37.5 (0) & 0.6 (0) & 0.8 (0) & 88.8 (0) & 0.6 (0) & 1.9 (0) \\ 
   &  &  & 1000 & 81.7 (0) & 0.5 (0) & 8.4 (0) & 99.5 (0) & 1.3 (0) & 16.5 (0) \\ 
   \hline
\end{tabular}
\endgroup
\end{table}

\begin{table}[t!]
\centering
\caption{For $(r,s) \in \{(3,3), (3,5), (5,5)\}$, rejection percentages of the goodness-of-fit test based on $S^{[n]}$ in~\eqref{eq:S:n} and Yule's coefficient computed from 1000 random samples of size $n \in \{100,500,1000\}$ generated, as explained in Section~\ref{sec:MC:gof}, from p.m.f.s whose copula p.m.f.\ is of the form~\eqref{eq:u:theta} with $C_\theta$ the Frank copula with a Kendall's tau in $\{0.33, 0.66\}$. The column `m' gives the marginal scenario. The integer between parentheses is the number of numerical issues encountered out of 1000 executions.} 
\label{tab:gof:3}
\begingroup\footnotesize
\begin{tabular}{rrrrrrrrrr}
  \hline
  \multicolumn{4}{c}{} & \multicolumn{3}{c}{$\tau=0.33$} & \multicolumn{3}{c}{$\tau=0.66$} \\ \cmidrule(lr){5-7} \cmidrule(lr){8-10} $r$ & $s$ & m & n & \multicolumn{1}{c}{Cl} & \multicolumn{1}{c}{GH} & \multicolumn{1}{c}{F} & \multicolumn{1}{c}{Cl} & \multicolumn{1}{c}{GH} & \multicolumn{1}{c}{F}  \\ \hline
3 & 3 & 1 & 100 & 13.7 (0) & 5.1 (0) & 4 (0) & 16 (0) & 1.4 (0) & 1.5 (0) \\ 
   &  &  & 500 & 61.7 (0) & 21 (0) & 4.3 (0) & 78.8 (0) & 12.5 (0) & 2.4 (0) \\ 
   &  &  & 1000 & 92 (0) & 38.6 (0) & 4.4 (0) & 99 (0) & 34.8 (0) & 3.5 (0) \\ 
   \\ &  & 2 & 100 & 12.7 (0) & 3.5 (0) & 3.1 (0) & 8.9 (0) & 1.8 (0) & 1.4 (0) \\ 
   &  &  & 500 & 49.2 (0) & 17.8 (0) & 5.1 (0) & 66 (0) & 8.9 (0) & 3.4 (0) \\ 
   &  &  & 1000 & 83.5 (0) & 32.7 (0) & 4.4 (0) & 97.2 (0) & 26.6 (0) & 4.2 (0) \\ 
   \\ &  & 3 & 100 & 7.8 (0) & 3.7 (0) & 1.6 (0) & 8.4 (0) & 1 (0) & 0.9 (0) \\ 
   &  &  & 500 & 58.3 (0) & 15.8 (0) & 5.1 (0) & 66.5 (0) & 5.4 (0) & 1 (0) \\ 
   &  &  & 1000 & 88.2 (0) & 38.5 (0) & 5 (0) & 97.4 (0) & 21.7 (0) & 1.7 (0) \\ 
   \\3 & 5 & 1 & 100 & 10.8 (0) & 3 (0) & 1 (0) & 19.9 (0) & 0.7 (0) & 0.1 (0) \\ 
   &  &  & 500 & 77.1 (0) & 29.8 (0) & 3.8 (0) & 96.6 (0) & 28.7 (0) & 2.4 (0) \\ 
   &  &  & 1000 & 98.8 (0) & 56.6 (0) & 4.9 (0) & 99.9 (0) & 64.2 (0) & 3.7 (0) \\ 
   \\ &  & 2 & 100 & 5.8 (0) & 1.3 (0) & 1.1 (0) & 16.8 (0) & 0.3 (0) & 0.1 (0) \\ 
   &  &  & 500 & 57.4 (0) & 18.9 (0) & 4.1 (0) & 81.5 (0) & 11.7 (0) & 1.4 (0) \\ 
   &  &  & 1000 & 89.7 (0) & 38.8 (0) & 6.5 (0) & 98.4 (0) & 44.6 (0) & 2.2 (0) \\ 
   \\ &  & 3 & 100 & 3 (0) & 0.6 (0) & 0.2 (0) & 2 (0) & 0 (0) & 0 (0) \\ 
   &  &  & 500 & 35.3 (0) & 10.5 (0) & 1.3 (0) & 68.7 (0) & 10.5 (0) & 1.3 (0) \\ 
   &  &  & 1000 & 71.8 (0) & 28.8 (0) & 3.5 (0) & 96.9 (0) & 29.3 (0) & 4 (0) \\ 
   \\5 & 5 & 1 & 100 & 5.1 (0) & 0.5 (0) & 0.6 (0) & 19.7 (0) & 0.4 (0) & 0.1 (0) \\ 
   &  &  & 500 & 82.9 (0) & 36.5 (0) & 5.6 (0) & 99.8 (0) & 38.6 (0) & 2.1 (0) \\ 
   &  &  & 1000 & 99.5 (0) & 71.8 (0) & 6 (0) & 100 (0) & 86.9 (0) & 2 (0) \\ 
   \\ &  & 2 & 100 & 1.2 (0) & 0.1 (0) & 0 (0) & 11 (0) & 0 (0) & 0 (0) \\ 
   &  &  & 500 & 53 (0) & 16 (0) & 4.2 (0) & 97.6 (0) & 17.7 (0) & 1.5 (0) \\ 
   &  &  & 1000 & 91.9 (0) & 46.4 (0) & 4.1 (0) & 100 (0) & 70 (0) & 1.5 (0) \\ 
   \\ &  & 3 & 100 & 0 (0) & 0 (0) & 0 (0) & 0.6 (0) & 0 (0) & 0 (0) \\ 
   &  &  & 500 & 12.2 (0) & 1.7 (0) & 0.2 (0) & 74.1 (0) & 2.9 (0) & 1.3 (0) \\ 
   &  &  & 1000 & 39.4 (0) & 10.1 (0) & 0.6 (0) & 99.2 (0) & 27.8 (0) & 1.7 (0) \\ 
   \hline
\end{tabular}
\endgroup
\end{table}

To study the finite-sample performance of the asymptotic chi-square goodness-of-fit tests described in the previous section, we consider similar data generating scenarios as in Section~\ref{sec:MC:est}. Table~\ref{tab:gof:1} (resp. Table~\ref{tab:gof:2}, Table~\ref{tab:gof:3}) reports rejection percentages of the test based on $S^{[n]}$ in~\eqref{eq:S:n} and Yule's coefficient computed from 1000 random samples of size $n \in \{100,500,1000\}$ generated from p.m.f.s whose copula p.m.f.\ is of the form~\eqref{eq:u:theta} with $C_\theta$ the Clayton (resp.\ Gumbel-Hougaard, Frank) copula with a Kendall's tau in $\{0.33, 0.66\}$ and whose margins are as in the marginal scenarios listed in Section~\ref{sec:MC:est}. In all the tables, the columns `Cl' (resp.\ `GH', `F') report rejection percentages when $\Jc$ in $H_0$ in~\eqref{eq:H0} is constructed from a Clayton (resp.\ Gumbel--Hougaard, Frank) copula. The integer between parentheses next to each rejection percentage is the number of numerical issues (either related to the convergence of the IPFP, numerical root finding or the necessary eigenvalue decomposition) out of 1000 executions. An inspection of the tables shows that such issues were very rarely encountered. In terms of rejection percentages, we see that the tests were never too liberal and that they tend to be too conservative for scenarios in which the probability of zero counts in the contingency tables of the generated samples is large. In a related way, the tests display generally good power when they are not overly conservative.

The lowest empirical levels and powers in Tables~\ref{tab:gof:1},~\ref{tab:gof:2} and~\ref{tab:gof:3} are frequently observed for the third marginal scenario. For instance, for that marginal scenario, $C_\theta$ the Clayton copula with a Kendall's tau of 0.66, $(r,s) = (5,5)$ and $n=500$, a realization of the empirical copula p.m.f.\ $u^{[n]}$ in~\eqref{eq:un} multiplied by $n$ and rounded to the nearest integer is:
$$
\begin{bmatrix}
80 &  16 &   3  &  0 &   0 \\
18 &  53 &  21  &  9 &   0 \\
2 &  18 &  40 &  27 &  12 \\
0 &   9 &  19 &  47 &  26 \\
0 &   3  & 17 &  17 &  62 \\
\end{bmatrix}.
$$
Note that, as already mentioned in Section~\ref{sec:MPL}, the latter could be interpreted as observed counts from the unknown copula p.m.f.\ $u$. The low counts in the lower-left and upper right-corners might be the reason for the conservative behavior of the goodness-of-fit tests based on $S^{[n]}$ in~\eqref{eq:S:n}. For that reason, in the next experiment, we focused on the $(r,s) = (5,5)$ case and the third marginal scenario, and considered groupings according to the following matrix:
\begin{equation}
  \label{eq:G55}
\begin{bmatrix}
 &  &  1 & 2 & 2 \\
 & &  1 & 2 & 2 \\
3 & 3 & & &  \\
4 & 4  &  &   & \\
4 & 4 &  &  &  \\
\end{bmatrix}
\end{equation}
where elements with the same integer are to be regrouped and from which we can form the 25 by 25 grouping matrix $G$ to be used in the statistic $S^{[n]}_G$ in~\eqref{eq:S:n:G}. Notice that the resulting groupings may not always guarantee that all aggregated counts are above 5. The latter will certainly not be true in general for $n=100$. The rejection percentages of the goodness-of-fit test based on $S_G^{[n]}$ and Yule's coefficient are reported in Table~\ref{tab:gof:G55}. A comparison with the horizontal blocks of Tables~\ref{tab:gof:1},~\ref{tab:gof:2} and~\ref{tab:gof:3} corresponding to $(r,s) = (5,5)$ and the third marginal scenario shows a clear improvement of the empirical levels and an increase of the powers when $n \geq 500$.

\begin{table}[t!]
\centering
\caption{For $(r,s) = (5,5)$, rejection percentages of the goodness-of-fit test based on $S_G^{[n]}$ in~\eqref{eq:S:n:G} and Yule's coefficient with $G$ formed according to~\eqref{eq:G55} computed from 1000 random samples of size $n \in \{100,500,1000,2000\}$ generated from a p.m.f.\ whose copula p.m.f.\ is of the form~\eqref{eq:u:theta} with $C_\theta$ the Clayton (Cl), Gumbel--Hougaard (GH) or Frank copula (F) with a Kendall's tau in $\{0.33, 0.66\}$ and whose margins are binomial with parameters 4 and 0.5.} 
\label{tab:gof:G55}
\begingroup\footnotesize
\begin{tabular}{rrrrrrrrrrr}
  \hline
  \multicolumn{2}{c}{} & \multicolumn{3}{c}{$C_\theta =$ Cl} & \multicolumn{3}{c}{$C_\theta = $ GH} & \multicolumn{3}{c}{$C_\theta =$ F} \\ \cmidrule(lr){3-5} \cmidrule(lr){6-8} \cmidrule(lr){9-11} $\tau$ & n & Cl & GH & F & Cl & GH & F & Cl & GH & F  \\ \hline
0.33 & 100 & 0.0 & 1.7 & 0.2 & 1.6 & 0.2 & 0.1 & 0.5 & 0.6 & 0.1 \\ 
   & 500 & 3.9 & 71.1 & 28.9 & 67.1 & 4.0 & 13.7 & 30.7 & 22.8 & 3.2 \\ 
   & 1000 & 4.0 & 97.3 & 67.9 & 95.3 & 3.6 & 35.7 & 61.1 & 43.5 & 5.6 \\ 
   & 2000 & 3.3 & 100.0 & 93.8 & 100.0 & 4.6 & 66.9 & 92.7 & 73.8 & 4.2 \\ 
   \\0.66 & 100 & 0.1 & 17.3 & 2.3 & 1.2 & 0.1 & 0.0 & 0.8 & 1.5 & 0.1 \\ 
   & 500 & 3.5 & 99.9 & 94.0 & 84.7 & 3.3 & 29.7 & 84.9 & 51.3 & 3.6 \\ 
   & 1000 & 3.6 & 100.0 & 100.0 & 99.3 & 4.4 & 78.4 & 99.7 & 88.6 & 5.3 \\ 
   & 2000 & 4.1 & 100.0 & 100.0 & 100.0 & 3.9 & 99.5 & 100.0 & 99.9 & 4.6 \\ 
   \hline
\end{tabular}
\endgroup
\end{table}

In a last experiment, we focused on the $(r,s) = (10,10)$ case and second marginal scenario listed in Section~\ref{sec:MC:est}, and decided to form the grouping matrix $G$ according to the following matrix:
\begin{equation}
  \label{eq:G1010}
\begin{bmatrix}
1  & 1 & 1 &  2 & 2 & 2& 2& 3& 3& 3\\
1  & 1 & 1 &  2& 2& 2& 2& 3& 3& 3\\
1  & 1 & 1 &  2& 2& 2& 2& 3& 3& 3\\
4  & 4 & 4 &  & & & &5 &5 &5 \\
4  & 4 & 4 &  & & & &5 &5 &5 \\
4  & 4 & 4 &  & & & &5 &5 &5 \\
4  & 4 & 4 &  & & & &5 &5 &5 \\
6  &6  &6  &7  &7 &7 &7 &8 &8 &8 \\
6  &6  &6  &7  &7 &7 &7 &8 &8 &8 \\
6  &6  &6  &7  &7 &7 &7 &8 &8 &8 \\
\end{bmatrix}.
\end{equation}
The rejection percentages are reported in Table~\ref{tab:gof:G1010} for $n \in \{500,1000,2000,4000\}$. The results seem to confirm that the goodness-of-fit tests based on $S^{[n]}_G$ in~\eqref{eq:S:n:G} can be well-behaved in many scenarios provided groupings are performed to avoid ``very low counts'' in the matrix $n u^{[n]}$.

\begin{table}[t!]
\centering
\caption{For $(r,s) = (10,10)$, rejection percentages of the goodness-of-fit test based on $S_G^{[n]}$ in~\eqref{eq:S:n:G} and Yule's coefficient with $G$ formed according to~\eqref{eq:G1010} computed from 1000 random samples of size $n \in \{500,1000,2000,4000\}$ generated from a p.m.f.\ whose copula p.m.f.\ is of the form~\eqref{eq:u:theta} with $C_\theta$ the Clayton (Cl), Gumbel--Hougaard (GH) or Frank copula (F) with a Kendall's tau in $\{0.33, 0.66\}$ and whose margins are as in the second marginal scenario of Section~\ref{sec:MC:est}.} 
\label{tab:gof:G1010}
\begingroup\footnotesize
\begin{tabular}{rrrrrrrrrrr}
  \hline
  \multicolumn{2}{c}{} & \multicolumn{3}{c}{$C_\theta =$ Cl} & \multicolumn{3}{c}{$C_\theta = $ GH} & \multicolumn{3}{c}{$C_\theta =$ F} \\ \cmidrule(lr){3-5} \cmidrule(lr){6-8} \cmidrule(lr){9-11} $\tau$ & n & Cl & GH & F & Cl & GH & F & Cl & GH & F  \\ \hline
0.33 & 500 & 6.0 & 66.4 & 29.8 & 73.7 & 8.6 & 17.8 & 35.6 & 11.1 & 7.7 \\ 
   & 1000 & 3.6 & 97.9 & 61.0 & 97.5 & 6.2 & 26.4 & 68.1 & 18.3 & 5.5 \\ 
   & 2000 & 6.1 & 100.0 & 94.9 & 100.0 & 6.4 & 45.9 & 95.8 & 38.7 & 5.5 \\ 
   & 4000 & 5.0 & 100.0 & 100.0 & 100.0 & 4.7 & 81.7 & 100.0 & 78.5 & 5.4 \\ 
   \\0.66 & 500 & 2.2 & 91.0 & 54.1 & 94.8 & 2.8 & 16.5 & 57.9 & 6.4 & 3.0 \\ 
   & 1000 & 3.3 & 100.0 & 98.1 & 99.9 & 3.8 & 39.8 & 97.4 & 20.7 & 4.9 \\ 
   & 2000 & 5.3 & 100.0 & 100.0 & 100.0 & 4.3 & 78.0 & 100.0 & 62.1 & 4.6 \\ 
   & 4000 & 4.3 & 100.0 & 100.0 & 100.0 & 5.3 & 98.4 & 100.0 & 97.1 & 4.1 \\ 
   \hline
\end{tabular}
\endgroup
\end{table}

\subsection{Chi-square tests based on a semi-parametric bootstrap}
\label{sec:GOF:pb}

When $(X_1,Y_1),\dots,(X_n,Y_n)$ is a random sample, the goodness-of-fit tests based $S^{[n]}_G$ in~\eqref{eq:S:n:G} can also be carried out using an appropriate adaptation of the so-called parametric bootstrap \citep[see, e.g.,][]{StuGonPre93,GenRem08}. Specifically, in our case, the latter could be called a semi-parametric bootstrap. The testing procedure that we propose is as follows:
\begin{enumerate}
\item For the hypothesized parametric family of copula p.m.f.s $\Jc$ in~\eqref{eq:Jc}, estimate $\theta_0$ by $\hat \theta^{[n]}$, where $\hat \theta^{[n]}$ is one the three method-of-moments estimators in~\eqref{eq:MoM:estimators} or the maximum pseudo-likelihood estimator in~\eqref{eq:theta:n}, and compute $S_G^{[n]}$ in~\eqref{eq:S:n:G}.

\item Let $p^{[n,1]}$ and $p^{[n,2]}$ denote the margins of $p^{[n]}$ in~\eqref{eq:pn}, and, using~\eqref{eq:I:proj:Frechet} with $a = p^{[n,1]}$ and $b = p^{[n,2]}$, form $p^{[n,\hat \theta^{[n]}]} = \Ic_{p^{[n,1]},p^{[n,2]}}(u^{[\hat \theta^{[n]}]})$, which is a consistent semi-parametric estimate of the p.m.f.\ of $(X,Y)$ under $H_0$ in~\eqref{eq:H0}.

\item For some large integer $M$, repeat the following steps for $l \in \{1,\dots,M\}$:
  \begin{enumerate}
  \item Generate a random sample $(X_1^{(l)},Y_1^{(l)}),\dots,(X_n^{(l)},Y_n^{(l)})$ from $p^{[n,\hat \theta^{[n]}]}$.
  \item From the generated sample, compute the version $p^{[n],l}$ of $p^{[n]}$ in~\eqref{eq:pn}, the version $u^{[n],l} = \Uc(p^{[n],l})$ of $u^{[n]}$, where $\Uc$ is defined in~\eqref{eq:I:proj:unif}, and the version $\hat \theta^{[n],l}$ of $\hat \theta^{[n]}$.
  \item Form an approximate realization under $H_0$ of $S_G^{[n]}$ in~\eqref{eq:S:n:G} as
    $$
    S_G^{[n],l} = n \| \diag(G \, \vect(u^{[\hat \theta^{[n],l}]}))^{-1/2} G \, \vect(u^{[n],l} - u^{[\hat \theta^{[n],l}]}) \|_2^2.
    $$
  \end{enumerate}
\item Finally, estimate the p-value of the test as
  $$
  \frac{1}{M} \sum_{l = 1}^M \1(S_G^{[n],l} \geq S_G^{[n]}).
  $$
\end{enumerate}

\begin{remark}
From a theoretical perspective, to attempt to show the null asymptotic validity of the proposed procedure, one could try to start from the work of \cite{GenRem08}. From an empirical perspective, studying its finite-sample behavior is unfortunately substantially more computationally costly than studying the finite-sample performance of the asymptotic test described in Section~\ref{sec:GOF:asym}. \qed
\end{remark}

\section{Data example}
\label{sec:data:example}

We briefly illustrate the proposed methodology on a bivariate data set initially analyzed in \cite{Goo79}. The underlying discrete random variables are $X$, the occupational status of a British male subject, and $Y$ the occupational status of the subject's father. The occupational status can take $r = s = 8$ different ordered values encoded as the first 8 integers. \cite{Goo79} does not explain why the underlying possible values are considered to be ordered and refers the reader to \cite{Mil60}, among others. We note that there seems to be no reasons to consider that certain values in $I_{r,s}$ cannot be realizations of $(X,Y)$. In other words, it seems meaningful to assume that the unknown p.m.f.\ $p$ of $(X,Y)$ has support $I_{r,s}$ and thus that the decomposition in~\eqref{eq:sklar:like} holds. The random vector $(X,Y)$ was observed on $n = 3498$ subjects. The resulting bivariate contingency table can be obtained in $\textsf{R}$ by typing \texttt{data(occupationalStatus)} and is reproduced below for convenience:
\begin{equation}
  \label{eq:cont:tab}
\begin{bmatrix}
  50 & 19 & 26 &  8 &  7 & 11 &  6 &  2 \\
  16 & 40 & 34 & 18 & 11 & 20 &  8 &  3 \\
  12 & 35 & 65 & 66 & 35 & 88 & 23 & 21 \\
  11 & 20 & 58 &110 & 40 &183 & 64 & 32 \\
  2  & 8 & 12&  23 & 25 & 46 & 28 & 12 \\
  12&  28 &102 &162 & 90& 554& 230& 177 \\
  0 &  6&  19 & 40 & 21 &158& 143 & 71 \\
  0 &  3 & 14 & 32 & 15& 126 & 91 &106 \\
\end{bmatrix}.
\end{equation}

\begin{figure}[t!]
\begin{center}
  \includegraphics*[width=0.4\linewidth]{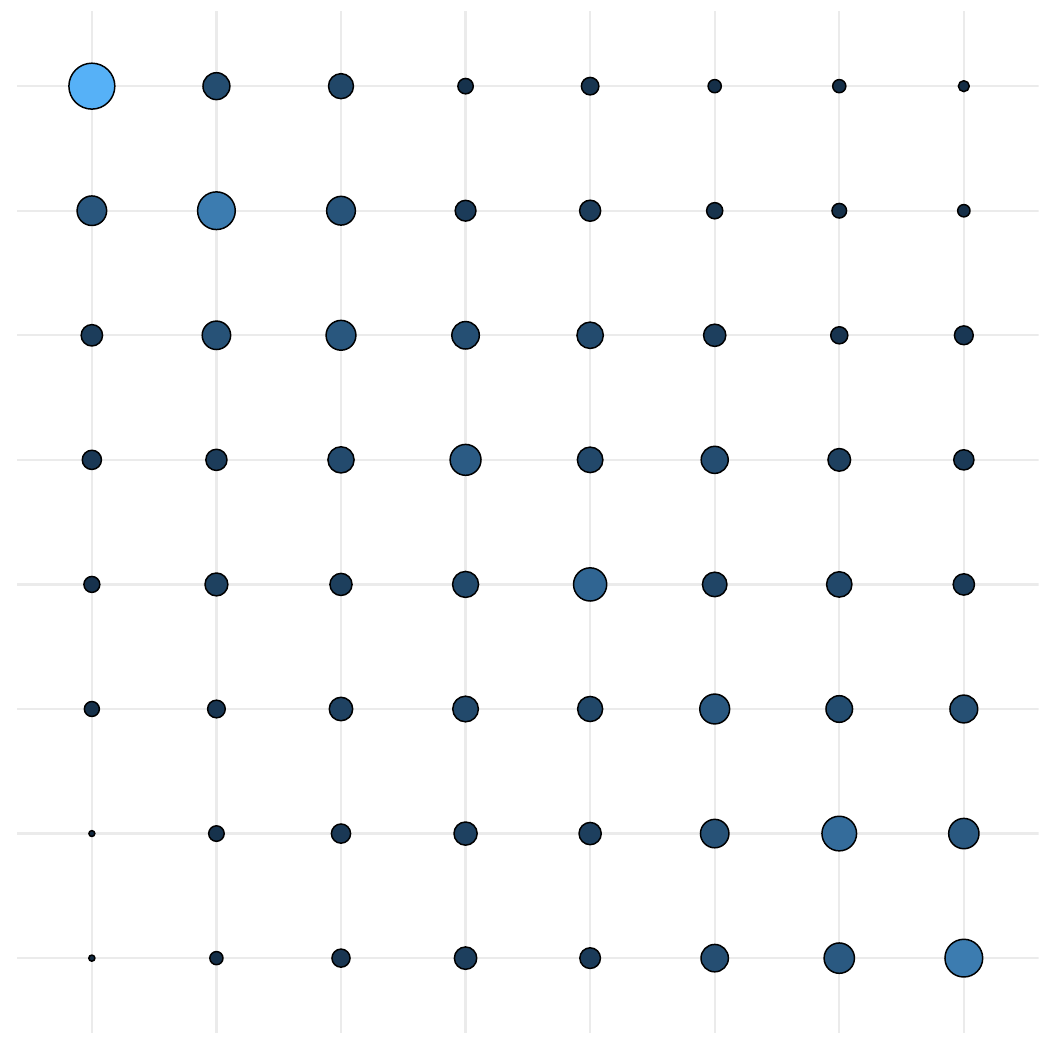}
  \caption{\label{fig:illus:un} Ballon plot of the empirical copula p.m.f.\ $u^{[n]}$ computed using~\eqref{eq:pn:q} and~\eqref{eq:un} from the contingency table given in~\eqref{eq:cont:tab}.}
\end{center}
\end{figure}

To estimate the unknown copula p.m.f.\ $u = \Uc(p)$ of $(X,Y)$, we first computed $p^{[n]}$ in~\eqref{eq:pn:q} from the above contingency table and then the empirical copula p.m.f.\ $u^{[n]}$ using~\eqref{eq:un}. The latter is represented in Figure~\ref{fig:illus:un} under the form of a \emph{ballon plot} created using the \textsf{R} package \texttt{ggpubr} \citep{ggpubr}. To complement Figure~\ref{fig:illus:un}, we also provide the matrix $n u^{[n]}$ rounded to the nearest integer,
\begin{equation}
  \label{eq:n:un}
  \begin{bmatrix}
    253 &  70&   58 &  14 &  21 &   8 &   8 &   4 \\
    88 & 160 &  82 &  35 &  36  & 17 &  12  &  7 \\
    38 &  80 &  90 &  74 &  66  & 42 &  20 &  26 \\
    28 &  37 &  65 &  99 &  61 &  71 &  44 &  32 \\
    16  & 46 &  42 &  64 & 118 &  55  & 60 &  37 \\
    13 &  22 &  48  & 62 &  58 &  91 &  68 &  76 \\
    0  & 15  & 28 &  47  & 42 &  80 & 130 &  94 \\
    0  &  8 &  24 &  43 &  34 &  73  & 95&  160 \\
  \end{bmatrix},
\end{equation}
which, as already mentioned, could be interpreted as observed counts from the unknown copula p.m.f.\ $u$. Estimates of Yule's coefficient $\rho$ in~\eqref{eq:Yule} as well as the gamma coefficient $\gamma$ and the tau coefficient $\tau$ in~\eqref{eq:gamma:tau} can be computed using~\eqref{eq:rho:gamma:tau:estimators} and are $\rho^{[n]} \simeq 0.63$, $\gamma^{[n]} \simeq 0.56$ and $\tau^{[n]} \simeq 0.5$, respectively. The latter seems to indicate a moderately strong dependence between $X$ and $Y$.

\begin{table}[t!]
\centering
\caption{Estimates of the parameter $\theta$ for families $\Jc$ in~\eqref{eq:Jc} constructed via~\eqref{eq:u:theta}, where $C_\theta$ is either a Clayton (Cl), Gumbel--Hougaard (GH), Frank (F), Plackett (P), survival Clayton (sCl), survival Gumbel--Hougaard (sGH) or survival Joe (sJ) copula. The first three columns give the three method-of-moment estimates defined in~\eqref{eq:MoM:estimators}. The fourth and fifth colums report the maximum pseudo-likelihood estimate in~\eqref{eq:theta:n} and the value of $- \bar L^{[n]}(\theta^{[n]})/n$, respectively, where $\bar L$ is defined in~\eqref{eq:PL}.} 
\label{tab:illus:fit}
\begingroup\small
\begin{tabular}{lrrrrr}
  \hline
  $C_\theta$ & $\theta_{\rho}^{[n]}$ & $\theta_{\gamma}^{[n]}$ & $\theta_{\tau}^{[n]}$ & $\theta^{[n]}$ & $- \bar L^{[n]}(\theta^{[n]})/n$ \\ \hline
Cl & 1.712 & 1.724 & 1.722 & 1.548 & 3.906 \\ 
  GH & 1.865 & 1.861 & 1.863 & 1.789 & 3.940 \\ 
  J & 2.591 & 2.596 & 2.592 & 2.070 & 3.998 \\ 
  F & 4.886 & 4.945 & 4.968 & 5.001 & 3.914 \\ 
  P & 9.147 & 8.915 & 8.961 & 8.717 & 3.909 \\ 
  sCl & 1.712 & 1.724 & 1.722 & 1.208 & 3.987 \\ 
  sGH & 1.865 & 1.861 & 1.863 & 1.861 & 3.896 \\ 
  sJ & 2.591 & 2.596 & 2.592 & 2.393 & 3.909 \\ 
   \hline
\end{tabular}
\endgroup
\end{table}

Let us next turn to the parametric modeling of $u$. We first fitted eight one-parameter families $\Jc$ in~\eqref{eq:Jc} constructed via~\eqref{eq:u:theta}, where $C_\theta$ is either a Clayton, Gumbel--Hougaard, Frank, Plackett, survival Clayton, survival Gumbel--Hougaard or survival Joe copula \citep[see, e.g.,][Section~2.5 for the definition of a survival copula]{HofKojMaeYan18}. Method-of-moments and maximum pseudo-likelihood estimates are given in Table~\ref{tab:illus:fit}, the last column of which also reports the scaled maximized negative log-pseudo-likelihood $- \bar L^{[n]}(\theta^{[n]})/n$,  where $\bar L$ is defined in~\eqref{eq:PL}. The latter seems to indicate that a model based on a survival Gumbel--Hougaard copula fits best. This may not be surprising since survival Gumbel--Hougaard copulas are lower-tail dependent and an inspection of Figure~\ref{fig:illus:un} and the ``observed counts'' in~\eqref{eq:n:un} seems to reveal such a ``lower-tail'' dependence in the upper left-corner.

\begin{table}[t!]
\centering
\caption{Results of the goodness-of-fit tests based on $S_G^{[n]}$ in~\eqref{eq:S:n:G} and Yule's coefficient with $G$ formed such that the four values in the lower-left and upper-right corners of the copula p.m.f.s are grouped. The hypothesized family $\Jc$ in~\eqref{eq:H0} is constructed via~\eqref{eq:u:theta}, where $C_\theta$ is either a Clayton (Cl), Gumbel--Hougaard (GH), Frank (F), Plackett (P), survival Clayton (sCl), survival Gumbel--Hougaard (sGH) or survival Joe (sJ) copula. The first row gives the values of $S_G^{[n]}$. The second (resp.\ third) row reports the p-values obtained via the asymptotic procedure (resp.\ semi-parametric bootstrap) described in Section~\ref{sec:GOF:asym} (resp.\ Section~\ref{sec:GOF:pb}) with $M=10^4$.} 
\label{tab:illus:gof}
\begingroup\small
\begin{tabular}{lrrrrrrrr}
  \hline
   & Cl & GH & J & F & P & sCl & sGH & sJ \\ \hline
$S^{[n]}_G$ &  260.4 &  523.1 & 1374.9 &  279.9 &  245.4 & 1200.2 &  156.4 &  303.0 \\ 
  Asymptotic & 0.000 & 0.000 & 0.000 & 0.000 & 0.000 & 0.000 & 0.004 & 0.000 \\ 
  Semi-p.\ boot.\ & 0.000 & 0.000 & 0.000 & 0.000 & 0.000 & 0.000 & 0.011 & 0.000 \\ 
   \hline
\end{tabular}
\endgroup
\end{table}

To further assess the fit of the aforementioned eight models, we carried out goodness-of-fit tests based $S_G^{[n]}$ in~\eqref{eq:S:n:G} and Yule's coefficient with $G$ formed such that, following~\eqref{eq:n:un}, the four values in the lower-left and upper-right corners of the copula p.m.f.s are grouped. P-values are given in Table~\ref{tab:illus:gof} and were computed using both the asymptotic procedure of Section~\ref{sec:GOF:asym} and the semi-parametric bootstrap of Section~\ref{sec:GOF:pb} with $M=10^4$. As one can see, all models are rejected, say, at the 2\% level, but the least-rejected one is the one based on the survival Gumbel--Hougaard copula family. A detailed inspection of Figure~\ref{fig:illus:un} and~\eqref{eq:n:un} seems to indicate the presence of asymmetry with respect to the diagonal in the unknown copula p.m.f.\ $u$. The latter suggests that models based on suitable non-exchangeable generalizations of the survival Gumbel--Hougaard copula family might provide a better fit. This is left for future research as explained in the forthcoming concluding remarks.

\section{Concluding remarks}

Inspired by the seminal work of \cite{Gee20}, we investigated a copula-like modeling approach for discrete bivariate distributions based on $I$-projections on Fréchet classes and the IPFP. The starting point of the investigated methodology is the copula-like decomposition of bivariate p.m.f.s stated in Proposition~\ref{prop:sklar:like}. Focusing our attention on discrete bivariate distributions with rectangular supports, we proposed nonparametric and parametric estimation procedures as well as goodness-of-fit tests for the underlying copula p.m.f. Related asymptotic results were provided thanks to a differentiability result for $I$-projections on Fréchet classes which can be of independent interest. Monte-Carlo experiments were carried out to study the finite-sample performance of some of the investigated inference procedures and the proposed methodology was finally illustrated on a data example.

We end this section by stating a few remarks:
\begin{itemize}

\item As already hinted at in the introduction, the assumption of rectangular support could be replaced by an alternative postulate for the support of~$p$ (based for instance on domain knowledge). Let us briefly describe the practical difficulties that would need to be overcome to adapt the statistical modeling methodology put forward in this work. Following Proposition~\ref{prop:sklar:like}, one would first need to verify that there exists a bivariate p.m.f.\ with uniform margins that has the same support as $p$. Then, a suitable smoothed estimator of $p$, playing the role of~\eqref{eq:pn:q}, would need to be proposed so that the corresponding empirical copula p.m.f.\ of the form of~\eqref{eq:un} could be computed in practice via the IPFP. Finally, one would need to come up with parametric models of the form of~\eqref{eq:Jc} whose support matches the one postulated for $p$. Such investigations are left for future work.

\item The proposed inference procedures rely on the initial smoothing in~\eqref{eq:pn:q} whose practical aim is to ensure that the IPFP will numerically converge with ``high probability''. Our proposal is arbitrary and alternative smoothing strategies would certainly need to be empirically investigated.

\item From the point of view of statistical practice, the asymptotic results stated in Section~\ref{sec:est} could be further used to obtain asymptotic confidence intervals for the various nonparametric and parametric estimates. Furthermore, with additional implementation work, fitting and goodness-of-fit testing could be extended to multiparameter parametric families $\Jc$ in~\eqref{eq:Jc}. This is left for a future project. The fact that we considered only families $\Jc$ built via~\eqref{eq:u:theta} is only a matter of practical convenience because of the amount of available implemented functionalities on copulas in \textsf{R}.  As already mentioned in Section~\ref{sec:par:est}, some other classes of parametric families $\Jc$ were proposed in Section~7 of \cite{Gee20}.

\item Finally, we believe that the proposed methodology can be made fully multivariate as the key underlying results related to $I$-projections are not restricted to the bivariate case. One would however still need to carefully verify that the fact that the IPFP is usually studied only in the bivariate case is merely due to notational convenience.

\end{itemize}

\section*{Acknowledgments}

The authors would like to warmly thank two anonymous Referees for their very constructive and insightful comments on an earlier version of this manuscript.

\appendix

\section{Proof of Proposition~\ref{prop:diff:I:proj}}
\label{proof:prop:diff:I:proj}

\begin{proof}[\bf Proof of Proposition~\ref{prop:diff:I:proj}]
  The proof is based on the implicit function theorem \citep[see, e.g.,][Theorem 17.6, p 450]{Fit09}. Before applying this result, we need to define the underlying function and verify a certain number of related assumptions.

Since it is assumed that $\Gamma_{a,b,T} \neq \emptyset$, from Proposition~\ref{prop:I:proj:diag}, for any $x \in \Gamma_T$, $\Ic_{a,b}(x) = y_x^* = \arginf_{y \in \Gamma_{a,b}} D(y \|x)$ with $\supp{(y_x^*)} = T$.  It follows that, for any $x \in \Gamma_T$, $\Ic_{a,b}(x) = \arginf_{y \in \Gamma_{a,b,T}} D(y \|x)$ and thus
\begin{equation}
  \label{eq:I:proj:open}
  \vect_A(\Ic_{a,b}(x)) = \arginf_{z \in \Lambda_{a,b,A,T}} H(z \| \vect_B(x)),
\end{equation}
where $H$ is defined in~\eqref{eq:H:def} and $\Lambda_{a,b,A,T}$ in~\eqref{eq:Lambda:ab:A:T} is, by assumption, an open subset of $\R^{|A|}$. Notice that, from~\eqref{eq:KL:divergence}, for any $(z,w) \in \Lambda_{a,b,A,T} \times \Lambda_{B,T}$,
\begin{equation}
  \label{eq:H:sum}
H(z \| w) = \sum_{(i,j) \in T} c(z)_{ij} \log \frac{c(z)_{ij}}{d(w)_{ij}}
\end{equation}
and $\Lambda_{B,T}$ in~\eqref{eq:Lambda:B:T} is an open subset of $\R^{|B|}$ under the considered assumptions on $B$. Next, let $F:\Lambda_{a,b,A,T} \times \Lambda_{B,T}  \to \R^{|A|}$ be defined by
$$
F(z,w) = \partial_1 H(z \| w), \qquad z \in \Lambda_{a,b,A,T} \subset (0,1)^{|A|}, w \in \Lambda_{B,T} \subset (0,1)^{|B|}.
$$
From~\eqref{eq:H:sum}, $F$ is differentiable at any $(z,w) \in \Lambda_{a,b,A,T} \times \Lambda_{B,T}$. Furthermore, from the definition of~$F$, \eqref{eq:I:proj:open} and first-order necessary optimality conditions, we have that
\begin{equation}
  \label{eq:F:0}
  F(\vect_A(\Ic_{a,b}(x)), \vect_B(x)) = 0_{\R^{|A|}}, \qquad \text{ for all } x \in \Gamma_T.
\end{equation}
Finally, to be able to apply the implicit function theorem, we shall verify that
\begin{equation}
  \label{eq:invertible}
  \det [ \partial_1 F(\vect_A(\Ic_{a,b}(x)), \vect_B(x)) ] \neq 0, \qquad \text{ for all } x \in \Gamma_T.
\end{equation}
Consider the set $\Lambda_{T,T} \subset (0,1)^{|T|}$ defined in~\eqref{eq:Lambda:B:T} with $B = T$ and let $\tilde D$ be the map from $\Lambda_{T,T} \times \Lambda_{T,T}$ to $[0,\infty)$ defined by $\tilde D(s \| s') = D(\vect^{-1}_T(s) \| \vect^{-1}_T(s'))$, $s,s' \in \Lambda_{T,T}$. Then, from~\eqref{eq:KL:divergence}, for any $s,s' \in \Lambda_{T,T}$,
$$
\tilde D(s \| s') = \sum_{i = 1}^{|T|} s_i \log \frac{s_i}{s_i'}.
$$
Lemma~\ref{lem:strong:convex} below then implies that, for any $r \in \Lambda_{T,T}$, $\tilde D(\cdot \| r)$ is strongly convex with constant 1 on $\Lambda_{T,T}$ in the sense of \citet[Section~2.1.3]{Nes04}. From Theorem~2.1.9 in the previous reference, the latter is equivalent to the fact that, for any $r \in \Lambda_{T,T}$, $s,s' \in \Lambda_{T,T}$ and $t \in [0,1]$,
\begin{equation}
  \label{eq:strong:convex}
  t \tilde D(s \| r) + (1 - t) \tilde D(s' \| r) \geq \tilde D(t s + (1-t) s' \| r) + t(1-t) \frac{1}{2} \| s - s' \|_2^2.
\end{equation}
From~\eqref{eq:H:def}, we further have that, for any $z \in \Lambda_{a,b,A,T}$ and $w \in \Lambda_{B,T}$,
$$
H(z \| w ) = D(c(z) || d(w)) = \tilde D(\vect_T \circ c(z) || \vect_T \circ d(w)).
$$
Fix $w \in  \Lambda_{B,T}$, $z, z' \in  \Lambda_{a,b,A,T}$ and $t \in [0,1]$. Using the fact that, by definition, $\vect_T \circ c$ is a linear map and~\eqref{eq:strong:convex}, we obtain
\begin{align*}
  t H(z \|  w)& + (1 - t) H(z' \| w)  - H(t z + (1-t) z' \| w) \\
  =&  t \tilde D(\vect_T \circ c(z) \| \vect_T \circ d(w)) + (1 - t) \tilde D(\vect_T \circ c(z') \| \vect_T \circ d(w)) \\
              &- \tilde D(\vect_T \circ c(t z + (1-t) z') \| \vect_T \circ d(w)) \\
  =&  t \tilde D(\vect_T \circ c(z) \| \vect_T \circ d(w)) + (1 - t) \tilde D(\vect_T \circ c(z') \| \vect_T \circ d(w))  \\
              &- \tilde D(t \vect_T \circ c(z) + (1-t) \vect_T \circ c(z') \| \vect_T \circ d(w)) \\
              \geq& t(1-t) \frac{1}{2} \| \vect_T \circ c(z') - \vect_T \circ c(z) \|_2^2 \geq t(1-t) \frac{1}{2} \| z' - z \|_2^2.
\end{align*}
Hence, by Theorem~2.1.9 in \cite{Nes04}, for any $w \in \Lambda_{B,T}$, $H(\cdot \| w )$ is strongly convex with constant 1 on $\Lambda_{a,b,A,T}$. From Theorem~2.1.11 in the same reference, the latter is equivalent to the fact that, for any $w \in \Lambda_{B,T}$ and $z \in \Lambda_{a,b,A,T}$, $\partial_1 \partial_1 H(z \| w) - I_{|A| \times |A|}$ is positive semi-definite, where $I_{|A| \times |A|}$ is the $|A| \times |A|$ identity matrix. This implies that, for any $w \in \Lambda_{B,T}$ and $z \in \Lambda_{a,b,A,T}$, $\partial_1 \partial_1 H(z \| w)$ is positive definite.
 Using the definition of $F$, the latter is equivalent to the fact that, for any $w \in \Lambda_{B,T}$ and $z \in \Lambda_{a,b,A,T}$, $\partial_1 F(z \| w)$ is positive definite, which implies~\eqref{eq:invertible}. 

Fix $x \in \Gamma_T$. We can now apply the implicit function theorem \citep[see, e.g.,][Theorem 17.6, p 450]{Fit09} to conclude that there exists a scalar $r > 0$ and a differentiable function $G:\Bc_r(x) \to (0,1)^{|A|}$, where $\Bc_r(x)$ is an open ball of radius $r$ centered at $\vect_B(x)$, such that,
whenever $\| z - \vect_A(\Ic_{a,b}(x)) \|_2 < r$, $\| w - \vect_B(x) \|_2 < r$ and $F(z,w) = 0$, then $G(w) = z$. The latter and~\eqref{eq:F:0} imply that, for any $x' \in \Gamma_T$ such that $\| \vect_A(\Ic_{a,b}(x')) - \vect_A(\Ic_{a,b}(x)) \|_2 < r$ and $\| \vect_B(x') - \vect_B(x) \|_2 < r$, we have $G(\vect_B(x')) = \vect_A(\Ic_{a,b}(x'))$. From the continuity of $\Ic_{a,b}$ stated in Lemma~\ref{lem:continuity} below, we thus obtain that $G(\vect_B(x')) = \vect_A(\Ic_{a,b}(x'))$ for $x'$ in a neighborhood of $x$, or, equivalently, that $G(w) = \vect_A \circ \Ic_{a,b} \circ d(w)$ for $w$ in a neighborhood of $\vect_B(x)$.

Another consequence of the implicit function theorem is that
$$
\partial_1 F(G(w), w) J_G(w)  + \partial_2 F(G(w), w)) = 0_{\R^{|A| \times |B|}} \text{ for all } w \in \Bc_r(x),
$$
where $J_G(w)$ is the Jacobian matrix of $G$ evaluated at $w$. From~\eqref{eq:invertible}, the previous centered display implies that
$$
J_G(w) = - [\partial_1 F(G(w), w)]^{-1} \partial_2 F(G(w), w))
$$
for all $w$ in a neighborhood of $\vect_B(x)$.
\end{proof}

\begin{lem}
  \label{lem:strong:convex}
For $k \in \N$, let $\Delta_k = \{u \in (0,1)^k: \sum_{i=1}^k u_i = 1 \}$. Furthermore, for any $u,v \in (0,1)^k$, let
\begin{equation*}
  \label{eq:KL:univ}
  \bar D(v \| u) = \sum_{i=1}^k v_i \log{\frac{v_i}{u_i}}.
\end{equation*}
Then, for any $u \in \Delta_k$, the function $\bar D(\cdot \| u)$ is strongly convex with constant $1$ on $\Delta_k$ in the sense of \citet[Section~2.1.3]{Nes04}.
\end{lem}

\begin{proof}
Fix $u \in \Delta_k$. According to \citet[Definition 2.1.2, p 63]{Nes04}, we need to prove that, for any $v, w \in \Delta_k$,
$$
\bar D(w \| u) \geq \bar D(v \| u) + \partial_1 \bar D(v \| u)^T(w-v) + \frac{1}{2} \| w - v \|_2^2.
$$
Note that, for any $v \in (0,1)^k$, $\partial_1 \bar D(v \| u) = (\log (v_1/u_1) + 1,\dots, \log (v_k/u_k) + 1)$ and thus that, for any $v, w \in \Delta_k$,
\begin{align*}
  \partial_1 \bar D(v \| u)^T(w-v) &= \sum_{i=1}^k (\log (v_i/u_i) + 1)(w_i - v_i) =  \sum_{i=1}^k w_i \log (v_i/u_i) - \bar D(v \| u) \\
                                   & =  \sum_{i=1}^k w_i \log \left(\frac{w_i v_i}{u_i w_i} \right) - \bar D(v \| u) = \bar D(w \| u) - \bar D(w \| v) - \bar D(v \| u).
\end{align*}
Hence,
\begin{align*}
  \bar D(w \| u) &= \bar D(v \| u) + \partial_1 \bar D(v \| u)^T(w-v) + \bar D(w \| v) \\
                 &\geq \bar D(v \| u) + \partial_1 \bar D(v \| u)^T(w-v) + \frac{1}{2} \| w - v \|_1^2,
\end{align*}
as a consequence of Pinsker's inequality \citep[see, e.g.,][Lemma 2.5, p 88]{Tsy08} and the fact that the total variation distance coincides with the $L_1$ distance in the considered finite setting. The proof is complete since $\|\cdot\|_1 \geq \|\cdot\|_2$.
\end{proof}

\begin{lem}
  \label{lem:continuity}
  For any $T \subset I_{r,s}$, $T \neq \emptyset$, such that there exists $y \in \Gamma_{a,b}$ with $\supp{(y)} \subset T$, the function $\Ic_{a,b}:\Gamma_T \to \Gamma_{a,b}$ is continuous.
\end{lem}

\begin{proof}
The result is a direct consequence of Theorem 3.3~(iii) in \citet{GieRef17}.
\end{proof}

\section{Proofs of Propositions~\ref{prop:sklar:like} and~\ref{prop:G:T}}
\label{proof:prop:sklar:like}

\begin{proof}[\bf Proof of Proposition~\ref{prop:sklar:like}]
  We first prove that Assertion~1 implies Assertion~2. Since there exists $v \in \Gamma_{\text{unif}}$ such that $\supp{(v)} = \supp{(p)}$, from Proposition~\ref{prop:I:proj}, the $I$-projection of $p$ on $\Gamma_{\text{unif}}$ exists and is unique, that is, $u = \Uc(p)$ exists and is unique. Furthermore, from Proposition~\ref{prop:I:proj:diag}, $u$ is diagonally equivalent to~$p$. Since there exists $q \in \Gamma_{p^{[1]}, p^{[2]}}$ such that $\supp{(q)} = \supp{(u)}$ (take $q = p$), $p' = \Ic_{p^{[1]},p^{[2]}}(u)$ exists and is unique, and $p'$ is diagonally equivalent to $u$. Since, by transitivity via $u$, $p'$ is diagonally equivalent to $p$, and since $p'$ and $p$ have the same margins, we can conclude from Property~1 of \cite{Pre80} \citep[see also][Lemma~27]{BroLeu18} that $p' = p$.

  We shall now prove that Assertion 2 implies Assertion 1. Since $u = \Uc(p) = \arginf_{y \in \Gamma_{\text{unif}}} D(y \| p)$, we have that $\supp{(u)} \subset \supp{(p)}$ by the definition of $D$ in~\eqref{eq:KL:divergence}. Similarly, $p = \Ic_{p^{[1]},p^{[2]}}(u) = \arginf_{y \in \Gamma_{p^{[1]},p^{[2]}}} D(y \| u)$ which implies that $\supp{(p)} \subset \supp{(u)}$. Assertion~1 thus holds since $u \in \Gamma_{\text{unif}}$ and $\supp{(u)} = \supp{(p)}$.
\end{proof}

\begin{proof}[\bf Proof of Proposition~\ref{prop:G:T}]
Let $(U,V)$ have p.m.f.\ $v \in \Gamma_{\text{unif}}$. Then, Goodman's and Kruskal's gamma \citep{GooKru54} and Kendall's tau~$b$ \citep[see, e.g.,][]{KenGib90} of $(U,V)$ are respectively defined by
$$
G(v) = \frac{\kappa(v) - \delta(v)}{\kappa(v) + \delta(v)} \qquad \text{and} \qquad T(v) = \frac{\kappa(v) - \delta(v)}{\sqrt{\Pr(U \neq U')\Pr(V \neq V')}},
$$
where $(U',V')$ is an independent copy of $(U,V)$ and
\begin{align*}
  \kappa(v) = \Pr\{(U-U')(V-V') > 0 \} \qquad \text{and} \qquad \delta(v) = \Pr\{(U-U')(V-V') < 0 \}.
\end{align*}
For $\kappa(v)$, we have that
\begin{align*}
  \kappa(v) &= \sum_{(i,j) \in I_{r,s}} \sum_{(i',j') \in I_{r,s}} \1 \{ (i - i') (j - j') > 0 \} \Pr(U = i, U' = i', V = j, V' = j') \\
  &= 2 \sum_{\substack{i \in \{1,\dots,r-1\} \\ j \in \{1,\dots,s-1\}}} \sum_{\substack{i' \in \{i+1,\dots,r\} \\ j' \in \{j+1,\dots,s\}}} v_{ij} v_{i'j'}.
\end{align*}
For $\delta(v)$, we can use the fact that $\delta(v) = \Pr\{(U-U')(V-V') \neq 0 \} - \kappa(v)$ and the fact that $\Pr\{(U-U')(V-V') \neq 0 \}$ can be expressed as
\begin{align*}
                        \sum_{(i,j) \in I_{r,s}} &\sum_{(i',j') \in I_{r,s}} \1 \{ (i - i') (j - j') \neq 0 \} \Pr(U = i, U' = i', V = j, V' = j') \\
                        &= \sum_{(i,j) \in I_{r,s}} \sum_{(i',j') \in I_{r,s}} \1 (i \neq i') \1(j \neq j') v_{ij} v_{i'j'} = \sum_{(i,j) \in I_{r,s}} v_{ij} \sum_{\substack{i' = 1 \\ i' \neq i}}^r \sum_{\substack{j' = 1 \\ j' \neq j}}^s v_{i'j'} \\
                        &= \sum_{(i,j) \in I_{r,s}} v_{ij} \sum_{\substack{i' = 1 \\ i' \neq i}}^r \left( \frac{1}{r} - v_{i'j} \right) = \sum_{(i,j) \in I_{r,s}} v_{ij}  \left( \frac{r-1}{r} - \sum_{\substack{i' = 1 \\ i' \neq i}}^r v_{i'j} \right) \\
                        &= \sum_{(i,j) \in I_{r,s}} v_{ij}  \left( \frac{r-1}{r} - \frac{1}{s} + v_{ij} \right) = 1 - \frac{1}{r} - \frac{1}{s} + \sum_{(i,j) \in I_{r,s}} v_{ij}^2.
\end{align*}
To obtain the expression $T(v)$, it remains to obtain the expressions of $\Pr(U \neq U')$ and $\Pr(V \neq V')$. We have
$$
\Pr(U \neq U') = \sum_{i,i' = 1}^r \1(i \neq i') \Pr(U = i, U' = i') = \frac{1}{r^2} \sum_{i=1}^r \sum_{\substack{i' = 1 \\ i' \neq i}}^r 1 = \frac{r-1}{r}
$$
and, similarly, $\Pr(V \neq V') = (s-1)/s$.
\end{proof}

\section{Proof of Proposition~\ref{prop:asym:un}}
\label{proof:prop:asym:un}

\begin{proof}[\bf Proof of Proposition~\ref{prop:asym:un}]
  We only prove the first claim as the other claims are immediate consequences of well-known results. We shall first apply Proposition~\ref{prop:diff:I:proj} with $a = u^{[1]}$, $b = u^{[2]}$, $T = I_{r,s}$, $A = I_{r-1} \times I_{s-1}$ and $B = I_{r,s} \setminus \{(r,s)\}$. Note that, with some abuse of notation, the map $d$ from $\Lambda_{B,T}$ in~\eqref{eq:Lambda:B:T} to $\Gamma_T$ in~\eqref{eq:Gamma:T} mentioned in its statement can be defined, for any $w \in \Lambda_{B,T}$, by
\begin{equation}
\label{eq:map:d}
d(w_{11}, w_{21}, \dots, w_{r-2,s}, w_{r-1,s}) = \begin{bmatrix}
w_{11}  & \dots & w_{1s} \\
\vdots &  & \vdots \\
w_{r1}  & \dots & 1 - \sum_{(i,j) \in B} w_{ij} \\
\end{bmatrix}
\end{equation}
and that, with some abuse of notation, the map $c$ from $\Lambda_{a,b,A,T}$ in~\eqref{eq:Lambda:ab:A:T} to $\Gamma_{a,b,T}$ in \eqref{eq:Gamma:ab:T} can be defined, for any $z \in \Lambda_{a,b,A,T}$, by
\begin{multline}
\label{eq:map:c}
c(z_{11}, z_{21}, \dots, z_{r-1,1}, \dots, z_{1,s-1}, z_{2,s-1}, \dots, z_{r-1,s-1}) \\ = \begin{bmatrix}
z_{11}  & \dots & z_{1,s-1} & 1/r - \sum_{j=1}^{s-1} z_{1j} \\
\vdots &  & \vdots & \vdots \\
z_{r-1,1}  & \dots & z_{r-1,s-1} & 1/r - \sum_{j=1}^{s-1} z_{r-1,j}  \\
1/s - \sum_{i=1}^{r-1} z_{i1} & \dots & 1/s - \sum_{i=1}^{r-1} z_{i,s-1} & 1/r + 1/s - 1 + \sum_{(i,j) \in A} z_{ij}
\end{bmatrix}.
\end{multline}

Recall the definition of $\Uc$ in~\eqref{eq:I:proj:unif} and let $u = \Uc(p)$. Proposition~\ref{prop:diff:I:proj} then implies that the map $\vect_A \circ \Uc \circ d$ is differentiable at $\vect_B(p)$ with Jacobian matrix at $\vect_B(p)$ equal to $- L_u^{-1} M_p$, where
$$
L_u = \partial_1 \partial_1 H(\vect_A(u) \| \vect_B(p)) \text{ and }  M_p = \partial_2  \partial_1 H(\vect_A(u) \| \vect_B(p)),
$$
the map $H$ is defined in~\eqref{eq:H:def} and, as we shall see below, the first (resp.\ second) matrix only depends on $u$ (resp.\ $p$). With some abuse of notation, let us denote the $(r-1)(s-1)$ components of $\partial_1 H(\vect_A(u) \| \vect_B(p))$ by
$$
(H'_{11},H'_{21},\dots,H'_{r-1,1},\dots,H'_{1,s-1}, H'_{2,s-1}, \dots, H'_{r-1,s-1}).
$$
Standard calculations give
$$
H'_{kl} = \log \left( \frac{c_{kl}(u)}{d_{kl}(p)} \right) - \log \left( \frac{c_{rl}(u)}{d_{rl}(p)} \right) - \log \left( \frac{c_{ks}(u)}{d_{ks}(p)} \right) + \log \left( \frac{c_{rs}(u)}{d_{rs}(p)} \right), \qquad (k,l) \in A,
$$
where $d_{kl}$ and $c_{kl}$ are the component maps of the maps $d$ and $c$ defined in~\eqref{eq:map:d} and~\eqref{eq:map:c}, respectively. Additional differentiation then leads to the expressions of the elements of the matrices $L_u$ and $M_p$ given in~\eqref{eq:Lu} and~\eqref{eq:Mp}, respectively.

Next, from the assumption that $\sqrt{n} (\hat p^{[n]} - p) \leadsto Z_p$ in $\R^{r \times s}$ and since
\begin{equation}
  \label{eq:asym:equiv}
  \sqrt{n} (p^{[n]} - \hat p^{[n]}) \p 0 \text{ in } \R^{r \times s},
\end{equation}
where $p^{[n]}$ is defined in~\eqref{eq:pn}, we immediately obtain that $\sqrt{n} (\vect_B(p^{[n]}) - \vect_B(p)) \leadsto \vect_B(Z_p)$ in $\R^{|B|}$, which, combined with the delta method \cite[see, e.g.,][Theorem 3.1]{van98} for the map $\vect_A \circ \Uc \circ d$ and, again,~\eqref{eq:asym:equiv}, implies that
\begin{multline*}
\sqrt{n} (\vect_A \circ \Uc \circ d \circ \vect_B(p^{[n]}) - \vect_A \circ \Uc \circ d \circ \vect_B(p)) \\ + L_u^{-1} M_p \sqrt{n} (\vect_B(\hat p^{[n]}) - \vect_B(p)) \p 0 \text{ in } \R^{|A|}.
\end{multline*}
Since, for the considered choices of $T$ and $B$, for any $y \in \Gamma_T$, $d \circ \vect_B(y) = y$, the latter is equivalent to
\begin{equation*}
\sqrt{n} \, \vect_A (u^{[n]} - u)  + L_u^{-1} M_p \sqrt{n} \, \vect_B(\hat p^{[n]} - p) \p 0  \text{ in } \R^{|A|}.
\end{equation*}
Notice from~\eqref{eq:map:c} that, for any $z \in \Lambda_{a,b,A,T}$, $c(z) = \vect^{-1}(K z) + C$, where $C$ is a constant $r \times s$ matrix. Hence, for any $y,y' \in  \Gamma_{a,b,T}$, $K \vect_A(y - y') = \vect(y-y')$ and, by the continuous mapping theorem, we obtain that
\begin{equation*}
\sqrt{n} \, \vect(u^{[n]} - u)  + K L_u^{-1} M_p \sqrt{n} \, \vect_B(\hat p^{[n]} - p) \p 0  \text{ in } \R^{rs}.
\end{equation*}
The desired result finally follows from the fact that $N \vect(y) = \vect_B(y)$, $y \in \R^{r \times s}$.
\end{proof}

\section{Proofs of the results of Section~\ref{sec:MPL}}
\label{proof:prop:MPL}

\begin{proof}[\bf Proof of Proposition~\ref{prop:MPL:consistency}]
To prove the consistency of $\theta^{[n]}$ , we shall use Theorem~5.7 in \cite{van98}. For any $\theta \in \Theta$, let $\ell^{[\theta]}_{ij} = \log u^{[\theta]}_{ij}$, $(i,j) \in I_{r,s}$, and let
\begin{equation}
  \label{eq:Mn:M}
  M^{[n]}(\theta) =  \frac{\bar L^{[n]}(\theta)}{n} = \sum_{(i,j) \in I_{r,s}} u^{[n]}_{ij} \ell^{[\theta]}_{ij} \quad \text{ and } \quad M(\theta) =  \sum_{(i,j) \in I_{r,s}} u_{ij} \ell^{[\theta]}_{ij}.
\end{equation}
From the definition of $\theta^{[n]}$ in~\eqref{eq:theta:n}, we have that $M^{[n]}(\theta^{[n]}) = \sup_{\theta \in \Theta} M^{[n]}(\theta)$, which implies that $M^{[n]}(\theta^{[n]}) \geq M^{[n]}(\theta_0)$. Furthermore, from the triangle inequality,
\begin{align*}
  \sup_{\theta \in \Theta} |M^{[n]}(\theta) - M(\theta) |  &\leq \sup_{\theta \in \Theta}  \sum_{(i,j) \in I_{r,s}} | \ell^{[\theta]}_{ij} (u^{[n]}_{ij} - u_{ij}) | \\ &= \sup_{\theta \in \Theta}  \sum_{(i,j) \in I_{r,s}} | \ell^{[\theta]}_{ij} | |u^{[n]}_{ij} - u_{ij}| \\
  &\leq | \log \lambda|   \sum_{(i,j) \in I_{r,s}} |u^{[n]}_{ij} - u_{ij}| \p 0 \text{ in } \R
\end{align*}
by Proposition~\ref{prop:consistency:un} and the continuous mapping theorem. Moreover, from the identifiability of the family $\Jc$ and Lemma~5.35 of \cite{van98}, we have that $\theta_0$ is the unique maximizer of $M$. The latter implies that, for every $\eps > 0$, $\sup_{\theta \in \Theta : \| \theta - \theta_0 \|_2 > \eps} M(\theta) < M(\theta_0)$. The consistency of $\theta^{[n]}$ finally follows from Theorem~5.7 in \cite{van98}.
\end{proof}

\begin{proof}[\bf Proof of Proposition~\ref{prop:MPL}]
We proceed along the lines of the proof of Theorem~5.21 in \cite{van98}. First, for any $\theta \in \Theta$, let
$$
\Psi^{[n]}_k(\theta) =  \frac{\partial M^{[n]}}{\partial \theta_k}(\theta), \,k \in I_m,
$$
where $M^{[n]}$ is defined in~\eqref{eq:Mn:M}, and
$$
\Psi^{[n]}(\theta) = \left(\Psi^{[n]}_1(\theta),\dots, \Psi^{[n]}_m(\theta) \right) = \sum_{(i,j) \in I_{r,s}} u^{[n]}_{ij} \dot \ell^{[\theta]}_{ij},
$$
where $\dot \ell^{[\theta]}_{ij}$ is defined in~\eqref{eq:dot:l:ij}. Since $\Theta$ is assumed to be open, the estimator $\theta^{[n]}$ in~\eqref{eq:theta:n} is a zero of $\Psi^{[n]}$. Also, let
\begin{equation}
\label{eq:Psi}
  \Psi(\theta) = \Ex_{\theta_0}(\dot \ell^{[\theta]}_{(U,V)}) = \sum_{(i,j) \in I_{r,s}} u_{ij} \dot \ell^{[\theta]}_{ij},
\end{equation}
where $(U,V)$ has p.m.f.\ $u$ and $\dot \ell^{[\theta]}_{ij}$ is defined in~\eqref{eq:dot:l:ij}. Then, from the discussion preceding the statement of Proposition~\ref{prop:MPL}, we have that $\theta_0$ is a zero of $\Psi$.

Since, for any $(i,j) \in I_{r,s}$, $\theta \mapsto \ell^{[\theta]}_{ij}$ is twice differentiable at any $\theta \in \Theta$, by the mean value theorem \citep[see, e.g.,][Theorem~15.29, p 408]{Fit09}, there exists $\eta > 0$ and a positive matrix $q \in \R^{r \times s}$ such that, whenever $\| \theta  - \theta_0 \|_2 < \eta$, for any $(i,j) \in I_{r,s}$,
\begin{equation}
  \label{eq:lipschitz}
\| \dot \ell^{[\theta]}_{ij} - \dot \ell^{[\theta_0]}_{ij} \|_2 \leq q_{ij} \| \theta - \theta_0 \|_2.
\end{equation}

Furthermore, from the triangle inequality and the inequality of Cauchy-Schwarz, we have
\begin{equation}
  \label{eq:Delta:n}
  \begin{split}
 \Delta^{[n]} &=  \left\| \sum_{(i,j) \in I_{r,s}} \dot \ell^{[\theta^{[n]}]}_{ij} \sqrt{n} ( u^{[n]}_{ij} - u_{ij} ) - \sum_{(i,j) \in I_{r,s}} \dot \ell^{[\theta_0]}_{ij} \sqrt{n} ( u^{[n]}_{ij} - u_{ij} ) \right\|_2^2 \\
 &\leq  \left\{ \sum_{(i,j) \in I_{r,s}} \| \dot \ell^{[\theta^{[n]}]}_{ij} - \dot \ell^{[\theta_0]}_{ij} \|_2 \times  | \sqrt{n} ( u^{[n]}_{ij} - u_{ij} ) |  \right\}^2 \leq A^{[n]} \times B^{[n]},
 \end{split}
\end{equation}
where
\begin{align}
  \label{eq:An:Bn}
  A^{[n]} =  \sum_{(i,j) \in I_{r,s}} \| \dot \ell^{[\theta^{[n]}]}_{ij} - \dot \ell^{[\theta_0]}_{ij} \|_2^2 \quad  \text{and} \quad B^{[n]} =  \sum_{(i,j) \in I_{r,s}} | \sqrt{n} ( u^{[n]}_{ij} - u_{ij} ) |^2.
\end{align}

To show that $\Delta^{[n]} \p 0$, let us first prove that $A^{[n]} \p 0$. Fix $\eps > 0$. Then, $\Pr( A^{[n]} > \eps) \leq I^{[n]} + J^{[n]}$, where
\begin{align*}
  I^{[n]} &= \Pr \left( \sum_{(i,j) \in I_{r,s}} \| \dot \ell^{[\theta^{[n]}]}_{ij} - \dot \ell^{[\theta_0]}_{ij} \|_2^2 > \eps, \| \theta^{[n]} - \theta_0 \|_2 < \eta \right), \\
  J^{[n]} &= \Pr \left( \| \theta^{[n]} - \theta_0 \|_2 \geq \eta \right).
\end{align*}
The fact that that $J^{[n]}$ converges to zero is a consequence of the consistency of~$\theta^{[n]}$ and the continuous mapping theorem. For $I^{[n]}$, using~\eqref{eq:lipschitz}, we obtain
\begin{align*}
  I^{[n]} &\leq \Pr \left(  \| \theta^{[n]} - \theta_0 \|_2^2 \sum_{(i,j) \in I_{r,s}} q_{ij}^2  > \eps, \| \theta^{[n]} - \theta_0 \|_2 < \eta \right) \\
  &\leq \Pr \left(  \| \theta^{[n]} - \theta_0 \|_2^2 \sum_{(i,j) \in I_{r,s}} q_{ij}^2  > \eps \right) \to 0,
\end{align*}
again as a consequence of the consistency of $\theta^{[n]}$. Hence, $A^{[n]} \p 0$.

Let us next verify that $\Delta^{[n]}$ in~\eqref{eq:Delta:n} converges in probability to zero. From Proposition~\ref{prop:asym:un}, we know that the weak convergence of $\sqrt{n} (\hat p^{[n]} - p)$ in $\R^{r \times s}$ implies the weak convergence of $\sqrt{n} (u^{[n]} - u)$ in $\R^{r \times s}$. Hence, from the continuous mapping theorem, $B^{[n]}$ in~\eqref{eq:An:Bn} is bounded in probability, which implies that $\Delta^{[n]}$ in~\eqref{eq:Delta:n} converges to zero in probability.

Moreover, since $\Psi^{[n]}(\theta^{[n]}) = 0$ and $\Psi(\theta_0) = 0$, we have that
\begin{equation}
  \begin{split}
  \label{eq:two}
  \sum_{(i,j) \in I_{r,s}} \dot \ell^{[\theta^{[n]}]}_{ij} \sqrt{n} ( u^{[n]}_{ij} - u_{ij} ) &= - \sum_{(i,j) \in I_{r,s}} \dot \ell^{[\theta^{[n]}]}_{ij} \sqrt{n} u_{ij}  \\
  &= \sum_{(i,j) \in I_{r,s}} \sqrt{n} ( \dot \ell^{[\theta_0]}_{ij} - \dot \ell^{[\theta^{[n]}]}_{ij} ) u_{ij}.
  \end{split}
\end{equation}
Since, for any $(i,j) \in I_{r,s}$, $\theta \mapsto \ell^{[\theta]}_{ij}$ is twice differentiable at any $\theta \in \Theta$, the map $\Psi$ in~\eqref{eq:Psi} is differentiable at its zero $\theta_0$ with Jacobian matrix at $\theta_0$ given by
$$
 \sum_{(i,j) \in I_{r,s}} u_{ij} \ddot \ell^{[\theta_0]}_{ij} = \Ex_{\theta_0}(\ddot \ell^{[\theta_0]}_{(U,V)}),
$$
where $\ddot \ell^{[\theta]}_{ij}$ is defined in~\eqref{eq:ddot:l:ij:matrix}. We then obtain from the delta method \citep[see, e.g.,][Theorem 3.1]{van98} that
\begin{equation*}
  \label{eq:three}
  \sqrt{n} \left( \sum_{(i,j) \in I_{r,s}} \dot \ell^{[\theta^{[n]}]}_{ij} u_{ij} - \sum_{(i,j) \in I_{r,s}} \dot \ell^{[\theta_0]}_{ij} u_{ij} \right) = \Ex_{\theta_0}(\ddot \ell^{[\theta_0]}_{(U,V)}) \sqrt{n} (\theta^{[n]} - \theta_0) + o_P(1).
\end{equation*}
Combining the latter display with~\eqref{eq:two} and the fact that $\Delta^{[n]}$ in~\eqref{eq:Delta:n} converges to zero in probability, we obtain that
$$
\Ex_{\theta_0}(\ddot \ell^{[\theta_0]}_{(U,V)})  \sqrt{n} (\theta^{[n]} - \theta_0) = - \sum_{(i,j) \in I_{r,s}} \dot \ell^{[\theta_0]}_{ij} \sqrt{n} ( u^{[n]}_{ij} - u_{ij} ) + o_P(1).
$$
The continuous mapping theorem,~\eqref{eq:Fisher} and Proposition~\ref{prop:asym:un} finally imply that
\begin{align*}
  \sqrt{n} (\theta^{[n]} - \theta_0) &= - \{ \Ex_{\theta_0}(\ddot \ell^{[\theta_0]}_{(U,V)}) \} ^{-1} \sum_{(i,j) \in I_{r,s}} \dot \ell^{[\theta_0]}_{ij} \sqrt{n} ( u^{[n]}_{ij} - u_{ij} ) + o_P(1), \\
                                     &= \{\Ex_{\theta_0} (\dot \ell^{[\theta_0]}_{(U,V)}\dot \ell^{[\theta_0],\top}_{(U,V)}) \} ^{-1}  \dot \ell^{[\theta_0]}  \sqrt{n} \, \vect ( u^{[n]} - u )  + o_P(1), \\
                                     &= \{\Ex_{\theta_0} (\dot \ell^{[\theta_0]}_{(U,V)}\dot \ell^{[\theta_0],\top}_{(U,V)}) \} ^{-1}  \dot \ell^{[\theta_0]}  J_{u,p} \sqrt{n} \, \vect(\hat p^{[n]} - p)  + o_P(1).
\end{align*}
\end{proof}

\section{Proofs of the results of Section~\ref{sec:GOF}}
\label{proofs:GOF}

\begin{proof}[\bf Proof of Proposition~\ref{prop:GOF}]
Under $H_0$ in~\eqref{eq:H0}, we can decompose the goodness-of-fit process as
\begin{align*}
  \sqrt{n} (u^{[n]} - u^{[\hat \theta^{[n]}]}) = \sqrt{n} (u^{[n]} - u) - \sqrt{n} (u^{[\hat \theta^{[n]}]} - u).
\end{align*}
From the delta method \cite[see, e.g.,][Theorem 3.1]{van98}, we then obtain that, under $H_0$,
$$
\sqrt{n} \, \vect(u^{[\hat \theta^{[n]}]} - u) = \sqrt{n} \, \vect(u^{[\hat \theta^{[n]}]} - u^{[\theta_0]}) = \dot u^{[\theta_0]} \sqrt{n} (\hat \theta^{[n]} - \theta_0) + o_P(1).
$$
From the assumptions, it follows that, under $H_0$, $\sqrt{n} \, \vect(u^{[n]} - u^{[\hat \theta^{[n]}]})$ has the same limiting distribution as
$$
\sqrt{n} \, \vect(u^{[n]} - u) - \dot u^{[\theta_0]} V_{u,p}^{[\theta_0]} \sqrt{n} \, \vect(\hat p^{[n]} - p).
$$
Finally, from Proposition~\ref{prop:asym:un}, we obtain that
\begin{align*}
  \sqrt{n} \, \vect(u^{[n]} - u^{[\hat \theta^{[n]}]}) &=  J_{u,p} \sqrt{n} \, \vect(\hat p^{[n]} - p) - \dot u^{[\theta_0]} V_{u,p}^{[\theta_0]} \sqrt{n} \, \vect(\hat p^{[n]} - p) + o_P(1).
\end{align*}
\end{proof}

\begin{proof}[\bf Proof of Proposition~\ref{prop:S:n:G}]
  From Proposition~\ref{prop:GOF}, the continuous mapping theorem and the fact that $u^{[\hat \theta^{[n]}]} \p u^{[\theta_0]}$ in $\R^{r \times s}$, we obtain that
  \begin{multline*}
    \diag(G \, \vect(u^{[\hat \theta^{[n]}]}))^{-1/2} \sqrt{n} \, G \, \vect(u^{[n]} - u^{[\hat \theta^{[n]}]}) \\
    =  \diag(G \, \vect(u^{[\theta_0]}))^{-1/2} G (J_{u,p} - \dot u^{[\theta_0]} V_{u,p}^{[\theta_0]}) \sqrt{n} \, \vect(\hat p^{[n]} - p) + o_P(1).
  \end{multline*}
  Consequently, when $(X_1,Y_1)$, \dots, $(X_n,Y_n)$ are independent copies of $(X,Y)$, the sequence
  $$
  \diag(G \, \vect(u^{[\hat \theta^{[n]}]}))^{-1/2} \sqrt{n} \, G \, \vect(u^{[n]} - u^{[\hat \theta^{[n]}]})
  $$
  is asymptotically centered normal with covariance matrix given by~\eqref{eq:cov}. The first claim is finally a consequence of Lemma~17.1 in \cite{van98}. The second claim immediately follows from the continuous mapping theorem.
\end{proof}

\bibliographystyle{chicago}
\bibliography{biblio}

\end{document}